\documentclass{myifcolog}
\usepackage{latexsym}
\usepackage{stmaryrd}
\newtheorem{theorem}{Theorem}
\newtheorem{corollary}{Corollary}
\newtheorem{lemma}{Lemma}
\newtheorem{definition}{Definition}

\newcommand{\eqdef}{\stackrel{\rm def}{=}}

\title{Projection-algebras and Quantum Logic}
\titlerunning{Projection-algebras}
\addauthor[daniel.lehmann1@mail.huji.ac.il]
{Daniel Lehmann}
{Selim and Rachel Benin School of Computer Science and Engineering \\ 
and Quantum Information Science Center \\ 
Hebrew University, Jerusalem 91904, Israel}
\authorrunning{Lehmann}
\titlethanks{Thanks to two anonymous referees for their remarks}

\begin{document}

\jnumber{}
\jvolume{11}
\jyear{2024}
\jreceived{1 March 2024}
\nopagenumber
\maketitle

\begin{abstract}
P-algebras are a non-commutative, non-associative generalization of Boolean algebras 
that are for quantum logic what Boolean algebras are for classical logic.
P-algebras have type \mbox{$ \langle X , 0 , ' , \cdot \rangle $} 
where $0$ is a constant, $'$ is unary and $\cdot$ is binary.
Elements of $X$ are called \emph{features}.
A partial order is defined on the set $X$ of features by \mbox{$ x \leq y $} iff 
\mbox{$ x \cdot y = x $}. 
Features commute, i.e., \mbox{$ x \cdot y = y \cdot x $} iff \mbox{$ x \cdot y \leq x $}.
Features $x$ and $y$ are said to be orthogonal iff \mbox{$ x \cdot y = 0 $} and orthogonality is a 
symmetric relation.
The operation $+$ is defined as the dual of $\cdot$ and it is commutative on orthogonal features.
The closed subspaces of a separable Hilbert space form a P-algebra under orthogonal 
complementation and projection of a subspace onto another one.
P-algebras are complemented orthomodular posets but they are not lattices.
Existence of least upper bounds for ascending sequences is equivalent to the existence 
of least upper bounds for countable sets of pairwise orthogonal elements.
Atomic algebras are defined and their main properties are studied.
The logic of P-algebras is then completely characterized.
The language contains a unary connective corresponding to the operation $'$ and a 
binary connective corresponding to the operation ``$\cdot$''.
It is a substructural logic of sequents where the Exchange rule is extremely limited.
It is proved to be sound and complete for P-algebras.
\end{abstract}

\section{Introduction} \label{sec:intro}
The purpose of this paper is to lay bare the algebraic and logical structures that allow physicists 
to model quantic systems.
It attempts to model the process of knowledge acquisition about a quantic system.
One gathers knowledge about such a system by a sequence of measurements: 
the result of a first measurement indicates that the system has a certain
feature $x$ and then the result of later measurement indicates that the system
has feature $y$.
The mathematical formalism used by physicists to describe this process is the following: 
$x$ and $y$ are closed subspaces of a separable Hilbert space and the result of the 
addition of the information brought by $y$ to that previously given by $x$ 
is the closed subspace that is the closure of the orthogonal projection of $x$ onto $y$.
One needs here to consider the closure of the projection because the projection of a 
closed subspace on a closed subspace is a subspace but not necessarily closed.
Our goal is to discover mathematical structures that are not as rich as Hilbert spaces in 
which one can model a good part of this knowledge acquisition process, 
the logical part of it.
The basic connective of quantum logic, conjunction, is, in this perspective, modeled 
by the projection operation. 
It is neither commutative nor associative.
Propositions are partially ordered by logical implication.
They form an orthomodular poset that is not a lattice.

Section~\ref{sec:motivation} describes the background behind the present effort.
Section~\ref{sec:P-algebras} defines and studies P-algebras, the algebraic structures 
we want to study.
Section~\ref{sec:syntax} defines the propositions and the sequents that form the 
language of quantum logic and their interpretation in P-algebras. 
Section~\ref{sec:rules} proposes a sound and complete deductive system for quantum 
logic.
Section~\ref{sec:further} considers further work and proposes a tentative conclusion.
Three appendices are devoted to more specific aspects of the study of P-algebras.
Appendix~\ref{sec:commuting} presents the properties of sets of commuting features.
Appendix~\ref{sec:CBP} considers infinitary properties of P-algebras and proves 
a fundamental result that equates the existence of l.u.b. for ascending sequences 
to the existence of l.u.b. for sets of pairwise orthogonal features.
Appendix~\ref{sec:atomic} studies atomic P-algebras.
Building on the result of those appendices upcoming work will propose richer structures
that include the probabilistic aspects of quantum systems.

\section{Motivation} \label{sec:motivation}
\subsection{Von Neumann's doubts about the Hilbert space formalism} \label{sec:doubts}
Three years after laying down the mathematical foundations of quantum physics 
in~\cite{vonNeumann:Quanten} John von Neumann wrote to Garret Birkhoff 
during the elaboration of~\cite{BirkvonNeu:36} 
(see~\cite{vNeumann_letters}, p. 59, letter dated Nov. 13, Wednesday, 1935):
\emph{I would like to make a confession
which may seem immoral: I do not believe absolutely in Hilbert space any more. 
After all Hilbert-space (as far as quantum-mechanical things are concerned) 
was obtained by generalizing Euclidean space, footing on the principle of 
``conserving the validity of all formal rules''. 
This is very clear, if you consider the axiomatic-geometric definition of
Hilbert-space, where one simply takes Weyl's axioms for a unitary-Euclidean
space, drops the condition on the existence of a finite linear basis, and
replaces it by a minimum of topological assumptions 
(completeness + separability). Thus Hilbert-space is the straightforward
generalization of Euclidean space, if one considers the \emph{ vectors}
as the essential notions. Now we begin to believe that it is not the 
\emph{ vectors} which matter but the \emph{ lattice of all linear (closed) 
subspaces}. Because:
\begin{enumerate}
\item The vectors ought to represent the physical \emph{states}, but they
do it redundantly, up to a complex factor only.
\item And besides the \emph{states} are merely a derived notion, the primitive
(phenomenologically given) notion being the \emph{qualities}, which correspond 
to the \emph{linear closed subspaces}.
\end{enumerate}
}

Indeed, modeling states of a quantic system by elements in a vector space requires 
justification since the essence of vectors is that two vectors can be added and that a 
vector can be multiplied by a scalar whether states cannot be added neither can they be 
multiplied by a scalar since they are unit vectors up to an arbitrary phase factor.
Von Neumann's doubts may be supported by the consideration of the role of linear 
operators in the formalization of QM.
Linear operators preserve scalar multiplication and addition of vectors, but if vectors are 
not the objects of study, scalar multiplication and sum have no phenomenological 
meaning, linear operators may not be the natural morphisms to consider.
Indeed, even though linear operators are bread and butter for quantum physicists:
\begin{itemize}
\item physicists have to consider operators that are not linear: the time reversal 
symmetry has to be represented by an antilinear operator, and
\item only two specific types of linear operators are in fact used by QM: self-adjoint
operators and unitary operators, but the centrality of self-adjoint operators does not 
follow from phenomenological principles.
\end{itemize}

\subsection{Background} \label{sec:background}
This paper's purpose is to propose a mathematical framework that is not a vector 
space. 
It builds on the results of~\cite{Lehmann-P-algebras-v1:2024}, itself 
following~\cite{LEG:Malg, Qsuperp:IJTP, SP:IJTP, Lehmann_andthen:JLC, Lehmann-metalinear1:2022}. 
The semantics of classical logic can be described by Boolean algebras.
This framework does not fit the fundamental novelties of quantum physics.
A vast literature followed G. Birkhoff and J. von Neumann's~\cite{BirkvonNeu:36} 
and proposed orthomodular lattices as the suitable semantic structures for the logic 
of quantum systems.

The main thesis of Garret Birkhoff and John von Neumann in~\cite{BirkvonNeu:36} 
is that quantic propositions represent closed subspaces of a separable Hilbert space.
In Section 6 they formulate the following postulate: 
``The set-theoretical product\footnote{i.e., intersection} of any two mathematical 
representatives of experimental propositions concerning a quantum-mechanical system, 
is itself the mathematical representative of an experimental proposition''.
They provide no justification for this postulate.
They go on explaining that the logical implication relation is modeled by subspace 
inclusion and is therefore a partial order relation.
Their postulate then implies that any two experimental propositions have a 
greatest lower bound, their intersection.
They conclude that the logical structures that model quantum logic are lattices.
But one may be reluctant to adopt the postulate mentioned above for the following reason. 
Consider a closed subspace $A$ that is an eigensubspace of a Hermitian operator $P$ 
and a closed subspace $B$ that is an eigensubspace of another Hermitian operator 
$Q$.
If the operators $P$ and $Q$ do not commute, it seems no experiment can result 
in the knowledge that the state of the system is in their intersection in view of 
Heisenberg's uncertainty principle.
This intersection does not correspond to any experimental proposition, and, therefore, 
the two propositions do not have a greatest lower bound in the implication partial 
ordering which cannot be a lattice ordering.

\section{P-algebras} \label{sec:P-algebras}
\subsection{Preliminaries} \label{sec:P-algebras-intro}
Our goal is to present a family of mathematical structures fit to represent
the features of a quantic system.
Our \emph{features} are the \emph{qualities} mentioned by von Neumann's quote 
in Section~\ref{sec:doubts}.
The reader should think of the closed subspaces of a separable Hilbert space and
the projections of one subspace on another one.
The main topic of this section is the study of the algebraic properties of the operation 
that projects a subspace of a Hilbert space onto another subspace.
The family of projection-algebras (P-algebras) to be defined in 
Definition~\ref{def:P-algebra} below extends the family of Boolean algebras.
As is the case for Boolean algebras, a large number of equivalent sets of properties 
may be proposed to define P-algebras.
Definition~\ref{def:P-algebra} presents one these sets, probably not the most 
elegant or the most economical.
In the wake of~\cite{PutnamHow:1974} where Hilary Putnam argues that logic is to 
Quantum Mechanics like Geometry to General Relativity, this paper has chosen to
present, as much as possible, P-algebras by properties of a logical nature.
The seven properties in Definition~\ref{def:P-algebra} provide a description of the
logical principles that are the foundation of reasoning about quantic systems.
The logical meaning of each of these properties will be described following the 
definition.
In P-algebras the binary operation ``$\cdot$'' is not assumed to be commutative: 
commutative P-algebras are exactly Boolean algebras, as shown in 
Corollary~\ref{the:commutative-B}.
Neither is it associative and the main challenge in Definition~\ref{def:P-algebra} is
to find the suitable weakenings for commutativity and associativity.

We shall consider algebras of type \mbox{$ \langle X , 0 , ' , \cdot \rangle $} 
where $0$ is a constant, $'$ is unary and ``$\cdot$" is binary.
The reader should think of $X$ as the set of all possible features of a quantic system
and of $0$ as the feature that no system possesses. 
If $x$ is a feature $x'$ is the feature that says the system does not exhibit feature $x$.
If $x$ and $y$ are features, $x \cdot y$ is the feature that expresses that feature 
$x$ has been assessed and (then) $y$ has been assessed. 
The operation ``$\cdot$'' in P-algebras, corresponds to the conjunction $\wedge$ of 
Boolean algebras.
Its dual operation $+$ that corresponds to disjunction $\vee$ will be discussed in 
Section~\ref{sec:+}.
Some notations will help us.
In items~\ref{bot-def} and~\ref{smile-def} below the order of the operands is not 
important since the relations $\bot$ and $\smile$ are commutative in P-algebras,
but, in item~\ref{leq-def}, the order is significant.

\begin{definition} \label{def:basic}
For any \mbox{$ x , y \in X $}
\begin{enumerate}
\item \label{1-def}
\mbox{$ 1 \eqdef 0' $}, 
\item \label{bot-def}
\mbox{$ x \, \bot \, y $} iff \mbox{$ x \cdot y = 0 $},
\item \label{leq-def}
\mbox{$ x \leq y $} iff \mbox{$ x \cdot y = x $}, 
\item \label{smile-def}
\mbox{$ x \smile y $} iff \mbox{$ x \cdot y \leq x $}.
\end{enumerate}
\end{definition}

Intuitively: 
\begin{itemize}
\item 
$1$ is the trivial feature, the feature every element possesses. 
\item
Feature $x$ is \emph{orthogonal} to $y$ iff $y$ cannot be observed after $x$.
In P-algebras this relation is symmetric, and orthogonality expresses that $x$ and $y$ 
are incompatible. 
\item
\mbox{$ x \leq y $} if feature $x$ implies feature $y$. 
This is expressed by the requirement that, given that feature $x$ has been assessed, 
the assessment of $y$ neither destroys $x$ nor any other feature possessed by 
the system.
\item
The relation $\smile$ expresses that $y$ does not disturb $x$. 
In P-algebras this relation is symmetric.
\mbox{$ x \smile y $} is equivalent to the claim that \mbox{$ x \cdot y = $} 
\mbox{$ y \cdot x $} or that \mbox{$ x \cdot y $}  is the g.l.b. of $x$ and $y$.
\end{itemize}

We want to compare the P-algebras that will be defined in 
Definition~\ref{def:P-algebra} with two other types of algebras.
First, to orthocomplemented lattices, when $0$ is interpreted 
as the bottom element $\bot$, 
$'$ as complementation and ``$\cdot$'' as greatest lower bound.
Secondly to Hilbert spaces, when $X$ is the set of closed subspaces of a separable 
Hilbert space $\cal H$, $0$ is the zero-dimensional subspace $\{\vec{0}\}$, 
$A'$ denotes the subspace orthogonal to $A$ and ``$\cdot$'' is projection, 
more precisely \mbox{$ A \cdot B $} is interpreted as \mbox{$ A \downarrow B $}, 
the closure of the projection of the closed subspace $A$ onto the closed subspace $B$.
\begin{itemize}
\item
In orthocomplemented lattices $1$ is the top element $\top$.
In Hilbert spaces it is the whole space $\cal H$.
\item
In orthocomplemented lattices, \mbox{$ x \, \bot \, y $} iff 
\mbox{$ g.l.b. ( x , y ) = \bot $}.
In Hilbert spaces, \mbox{$ A \, \bot \, B $} iff each of the elements of $A$ is orthogonal 
to each of the elements of $B$.
\item
In orthocomplemented lattices the $\leq$ symbol has its usual interpretation.
In Hilbert spaces \mbox{$ A \downarrow B = A $} iff \mbox{$ A \subseteq B $}.
\item
In orthocomplemented lattices, since the g.l.b. operation is commutative, we have
\mbox{$ x \smile y $} for any elements $x$ and $y$.
In Hilbert spaces, \mbox{$ A \downarrow B \subseteq A $} iff $A$ is the subspace 
spanned by some subspace of $B$ and some subspace orthogonal to $B$ and this is
equivalent to \emph{the projections $p_{A}$ and $p_{B}$ on $A$ and $B$ 
respectively, commute}.
\end{itemize}

\subsection{Definition of P-algebras} \label{sec:P-algebras-def}
Definition~\ref{def:P-algebra} defines P-algebras by a set of seven conditions.
Those conditions are not claimed to be independent or to be the simplest possible.

\begin{definition} \label{def:P-algebra}
Consider a structure \mbox{$ P = \langle X , 0 , ' , \cdot \rangle $}.
The structure $P$ is a \emph{P-algebra} iff it satisfies 
the following properties for any \mbox{$ x , y , z \in X $}:
\begin{enumerate}
\item \label{partial-order}
{\bf Partial order} \ the relation $\leq$ is a partial order, i.e., it is reflexive, 
anti-symmetric and transitive,
\item \label{P-commutativity}
{\bf P-commutativity} \ the relation $\smile$ is symmetric: if \mbox{$ x \smile y $}, 
then \mbox{$ y \smile x $},
\pagebreak[2]
\item \label{P-associativity}
{\bf P-associativity} 
\begin{enumerate}
\item \label{xyz0}
\mbox{$ ( x \cdot y ) \cdot z = 0 $} iff \mbox{$ ( z \cdot y ) \cdot x = 0 $}, 
\item \label{large-small-left}
if \mbox{$ x \leq y $}, then \mbox{$ ( x \cdot y ) \cdot z = $} 
\mbox{$ x \cdot ( y \cdot z ) $},
\item \label{large-small-right}
if \mbox{$ x \leq y $}, then \mbox{$ ( z \cdot y ) \cdot x = $}
\mbox{$ z \cdot ( y \cdot x ) $} ,
\end{enumerate}
\item \label{monotonicity}
{\bf Dot-monotonicity}
\begin{enumerate}
\item \label{left-monotonicity}
if \mbox{$ x \leq y $}, then \mbox{$ x \cdot z \leq y \cdot z $},
\item \label{right-monotonicity}
\mbox{$ x \cdot y \leq y $},
\end{enumerate}
\item \label{left-zero}
{\bf Z} \ \mbox{$ 0 \cdot x = 0 $}, equivalently \mbox{$ 0 \leq x $}, 
\item \label{complements}
{\bf Comp} \ \mbox{$ x \cdot x' = 0 $},
\item \label{O} 
{\bf O} \ if \mbox{$ x \cdot y \leq z $} and 
\mbox{$ x \cdot y' \leq z $}, 
then \mbox{$ x \leq z $}.
\end{enumerate}
\end{definition}

The essential difference between quantic logic and classical logic is that the operation 
``$\cdot$'' is not commutative.
Assessing that a particle has position $x$ and then assessing that it has momentum
$b$ leaves the physicist in a state of knowledge very different from the one he/she 
would be in after assessing, first, momentum $b$ and then position $x$.

{\bf Partial order} requires that implication be reflexive, anti-symmetric 
and transitive. 
This seems to be an unavoidable logical requirement: 
\begin{itemize}
\item any feature implies itself,
\item if two features $x$ and $y$ imply each other they are always seen together and
therefore cannot be distinguished,
\item if $x$ implies $y$ and $y$ implies $z$ then every time we see $x$ we have $z$ 
and $x$ implies $z$.
\end{itemize}
But note that we do \emph{not} require that the structure 
\mbox{$ \langle X , \leq \rangle $} be a \emph{lattice}.
In~\cite{Finch_lattice:69} Finch shows that four conditions on the operation ``$\cdot$'' 
imply a lattice structure. 
P-algebras satisfy the first three conditions, but not the fourth.
In orthocomplemented lattices condition {\bf Partial order} holds by assumption.
In Hilbert spaces it holds because inclusion is a partial order.

{\bf P-commutativity} is a fundamental principle for quantum logic:
the relation~$\smile$ expressing non-interference is symmetric: 
if the observation of feature $x$ cannot perturb feature $y$, 
then $y$ cannot perturb $x$.
We shall see that this relation is reflexive but it is not, in general, transitive.
In orthocomplemented lattices {\bf P-commutativity} holds vacuously since any two 
elements are in the relation $\smile$.
In Hilbert spaces, the condition holds since if $A$ is the subspace spanned by the union 
of a subspace of $B$ and a subspace orthogonal to $B$, 
\mbox{$ A = \overline{ B_{1} \cup B_{2}} $} where \mbox{$ B_{1} \subseteq B $}
and \mbox{$ B_{2} \subseteq B^\bot $}, one has 
\mbox{$ B \downarrow A = B_{1} $} and 
\mbox{$ B \downarrow A^\bot = B \cap B_{1}^{\bot} $} and therefore
\mbox{$ \overline{ ( B \downarrow A ) \cup ( B \downarrow A^\bot ) } \subseteq B $}.
The converse inclusion is obvious. 
Therefore $B$ is the subspace spanned by the union of a subspace of $A$ and a 
subspace orthogonal to $A$.

As was noted in~\cite{Lehmann_andthen:JLC} and then in~\cite{Fazio+3:2021}, 
projection is not associative.
Nonassociative logics have also been considered in the context of the Lambek calculus
in~\cite{Galatos+1:2009,Moot+1:2012,Cintula+2:2013,Buszkowski:2017}.
{\bf P-associativity} is a restricted associativity property. 
{\bf P-associativity} requires associativity only in three specific situations.
The first one, \ref{xyz0} requires that one of two expressions is equal to zero.
Note that the second expression is not \mbox{$ x \cdot ( y \cdot z ) $} but can be 
obtained from it by changing the order of the operands.
The meaning of such a requirement is that, if the sequence of measurements $x$, $y$, 
$z$ is impossible, then the opposite sequence $z$, $y$ , $x$ is also impossible.  
Condition~\ref{xyz0} expresses a property of invariance under time reversal that is 
of a quasi-logical nature.
In orthocomplemented lattices this property follows from associativity and 
commutativity of the g.l.b. operation.
In Hilbert spaces, \mbox{$ ( A \downarrow B ) \, \bot \, C $} iff 
\mbox{$ A \, \bot \, ( C \downarrow B ) $}. 
Condition~\ref{large-small-left} says that if $x$ implies $y$ measuring $z$ after having measured $x$ is not different from measuring $y \cdot z$ after $x$.
In orthocomplemented lattices, the property holds by associativity.
In Hilbert spaces \mbox{$ A \downarrow C = $} 
\mbox{$ A \downarrow ( B \downarrow C ) $} if \mbox{$ A \subseteq B $}.
Condition~\ref{large-small-right} says that, in any context, 
measuring a weak property and
then a stronger property is equivalent to measuring directly the latter.
Again, in orthocomplemented lattices, the property holds by associativity.
In Hilbert spaces it holds since the projection of a vector on a subspace can be 
obtained by projecting the vector first on a larger subspace and then projecting the
result on the subspace.
What {\bf P-associativity} does not imply is extremely interesting.
It does not imply that $ x \cdot y $ is a greatest lower bound for $x$ and $y$ 
in the partial order $\leq$, since we do not have, in general, 
\mbox{$ x \cdot y \leq x $}.

The conditions in {\bf Dot-monotonicity} are monotonicity properties for the operation 
``$\cdot$''.
Condition~\ref{left-monotonicity} expresses the fact that starting with more 
knowledge cannot result in less knowledge. 
Condition~\ref{right-monotonicity} says that, whatever our knowledge 
of a system is, if the feature $y$ is discovered, then the system possesses feature $y$.
In orthocomplemented lattices, {\bf Dot-monotonicity} follows from the definition of g.l.b.
In Hilbert spaces, it follows from the definition of the $\downarrow$ operation.

Condition {\bf Z} characterizes $0$ as the feature that implies any 
feature, a logical contradiction. 
It expresses the principle \emph{ex falso quodlibet} that holds both in orthocomplemented
lattices, since $0$ is the bottom element and in Hilbert space since the zero-dimensional
subspace is included in any subspace.

Condition {\bf Comp} characterizes the feature $x'$ as the feature 
that \emph{is incompatible with $x$}: it cannot be observed after $x$ has been 
observed.
It holds in orthocomplemented lattices by assumption and in Hilbert spaces by the 
definition of projection.

Condition {\bf O} expresses the core of the superposition principle. 
If a feature is found in two orthogonal contexts, then it does not 
depend on the context.
In orthocomplemented lattices {\bf O} holds by commutativity and distributivity.
In Hilbert spaces, for every vector $\vec{a}$ of $x$, 
\mbox{$ \vec{a} = $} \mbox{$ \vec{b} + \vec{c} $} where  
\mbox{$ \vec{b} \in x \cdot z$} and \mbox{$ \vec{c} \in x \cdot z'$}.
If \mbox{$ \vec{a_{y}}, \vec{b_{y}}, \vec{c_{y}} $} are respectively the projections 
of \mbox{$ \vec{a}, \vec{b}, \vec{c} $} on $y$,  we have 
\mbox{$ \vec{a_{y}} = $} \mbox{$ \vec{b_{y}} + \vec{c_{y}} $}.
But $\vec{b_{y}}$ is the projection of $\vec{b}$ on $z \cdot y$ since 
\mbox{$ \vec{b} \in z $} and, similarly,  $\vec{c_{y}}$ is the projection of $\vec{c}$
on $z' \cdot y$. 
If both $\vec{b_{y}}$ and $\vec{c_{y}}$ are in $w$ so is $\vec{a}$.

\subsection{Properties of P-algebras} \label{sec:P-algebras-ppties}
Our first result shows that the relation $\smile$ characterizes commuting features 
(item~\ref{smile-com}), that orthogonality is a symmetric relation 
(item~\ref{z-commutation}) and that orthogonality is intimately linked with the unary 
operation $'$ (item~\ref{bot'}).
\begin{theorem} \label{the:basic1}
In a P-algebra, for any \mbox{$ x , y , z \in X $},
\begin{enumerate}
\item \label{newzxy}
If \mbox{$ z \leq x $} and \mbox{$ z \leq y $}, then, \mbox{$ z \leq x \cdot y $}.
\item \label{smile-glb}
\mbox{$ x \smile y $} iff \mbox{$ x \cdot y $} is the g.l.b. of $x$ and $y$.
\item \label{smile-com}
\mbox{$ x \smile y $} iff \mbox{$ x \cdot y = y \cdot x $}.
\item \label{leq-sim}
If \mbox{$ x \leq y $}, then \mbox{$ x \smile y $}.
\item \label{small-large}
If \mbox{$ x \leq y $}, then \mbox{$ x \cdot y = x = y \cdot x $}.
\item \label{xyz<=yz}
\mbox{$ ( x \cdot y ) \cdot z \leq y \cdot z $}.
\item \label{x0}
\mbox{$ x \cdot 0 = 0 $}.
\item \label{z-commutation}
the relation $\bot$ is symmetric: if \mbox{$ x \cdot y = 0 $}, 
 then \mbox{$ y \cdot x = 0 $}.
Therefore \mbox{$ x' \cdot x = 0 $}.
\item \label{leq-bot}
if \mbox{$ x \, \bot \, y $} and \mbox{$ z \leq y $}, then \mbox{$ x \, \bot \, z $}.
\item \label{bot'}
\mbox{$ x \, \bot \, y $} iff \mbox{$ x \leq y' $}.
\end{enumerate}
\end{theorem}
\begin{proof}{\ }
\begin{enumerate}
\item 
By assumption, we have, \mbox{$ z \cdot x = z $} and \mbox{$ z \cdot y = z $}.
By {\bf P-associativity}, item~\ref{large-small-left}, 
\[
z = z \cdot y = ( z \cdot x ) \cdot y =  z \cdot ( x \cdot y ) 
\] 
and therefore \mbox{$ z \leq x \cdot y $}.
\item 
By Definition~\ref{def:P-algebra}, item~\ref{right-monotonicity}, one has 
\mbox{$ x \cdot y \leq y $}.
If \mbox{$ x \smile y $}, we have \mbox{$ x \cdot y \leq x $} and 
\mbox{$ x \cdot y $} is a lower bound for $x$ and $y$.
Item~\ref{newzxy} just above shows that it is their greatest lower bound.
The \emph{if} part is obvious.
\item 
Assume \mbox{$ x \smile y $}.
By Definition~\ref{def:P-algebra}, item~\ref{P-commutativity}, we have 
\mbox{$ y \smile x $}.
By item~\ref{smile-glb} just above, we have \mbox{$ x \cdot y = g.l.b. ( x , y ) $} and
\mbox{$ y \cdot x = g.l.b. ( y , x ) $}.
We conclude that \mbox{$ x \cdot y = y \cdot x $}.
Suppose, now, that \mbox{$ x \cdot y = y \cdot x $}.
By {\bf Dot-monotonicity}, \ref{right-monotonicity}, 
\mbox{$ x \cdot y = $} \mbox{$ y \cdot x \leq x $} 
and \mbox{$ x \smile y $}.
\item 
By assumption \mbox{$ x \cdot y = x $}. Conclude by reflexivity of $\leq$.
\item 
Assume \mbox{$ x \leq y $}.
By item~\ref{leq-sim}, \mbox{$ x \smile y $}.
By item~\ref{smile-com}, \mbox{$ x \cdot y = y \cdot x $} and, 
by Definition~\ref{def:basic}, \mbox{$ x \cdot y = x $}.
\item 
By {\bf Dot-monotonicity}, item~\ref{right-monotonicity}, \mbox{$ x \cdot y \leq y $}.
By {\bf Dot-monotonicity}, item~\ref{left-monotonicity}, 
\mbox{$ ( x \cdot y ) \cdot z \leq $} \mbox{$ y \cdot z $}.
\item 
By {\bf Z} and item~\ref{leq-sim} above, \mbox{$ 0 \smile x $}.
By item~\ref{smile-com}, \mbox{$ x \cdot 0 = 0 \cdot x $}.
Conclude by {\bf Z}.
\item 
Assume \mbox{$ x \cdot y = 0 $}.
By {\bf Z}, \mbox{$ x \cdot y \leq x $}, i.e., \mbox{$ x \smile y $} 
and by item~\ref{x0} just above \mbox{$ x \cdot y = y \cdot x $}. 
Therefore \mbox{$ y \cdot x = 0 $}.
Notice, now, that, by {\bf Comp}, \mbox{$ x \cdot x' = 0 $} and therefore 
\mbox{$ x' \cdot x = 0 $}. 
\item 
By assumption \mbox{$ x \cdot y = 0 $} and, by item~\ref{z-commutation}, 
\mbox{$ y \cdot x = 0 $}.
But \mbox{$ z \leq y $} and, by {\bf Dot-monotonicity}, item~\ref{left-monotonicity}, 
we have \mbox{$ z \cdot x \leq $} \mbox{$ y \cdot x = 0 $}.
\item 
Suppose \mbox{$ x \leq y' $}.
By {\bf Dot-monotonicity}, item~\ref{left-monotonicity}, we have
\mbox{$ x \cdot y \leq y' \cdot y $} and, by item~\ref{z-commutation} above
\mbox{$ x \cdot y = 0 $}.
Suppose, now, that \mbox{$ x \cdot y = 0 $}.
We have \mbox{$ x \cdot y \leq x \cdot y' $}.
Since, by reflexivity, \mbox{$ x \cdot y' \leq x \cdot y' $}, we can use {\bf O}, 
to obtain \mbox{$ x \leq x \cdot y' $}.
By {\bf Dot-monotonicity}, item~\ref{right-monotonicity} \mbox{$ x \cdot y' \leq y' $} 
and since, by {\bf Partial-order}, the relation $\leq$ is transitive, we have 
\mbox{$ x \leq y' $}.
\end{enumerate}
\end{proof}

A corollary of item~\ref{newzxy} is that any commutative P-algebra is a Boolean
algebra.
Such a result is well in line with the idea that the relation between Quantum Logic and
Classical Logic should be similar to the relation between Quantum Physics and Classical 
Physics, since it is common wisdom that Quantum Mechanics boils down to Classical 
Mechanics when all operators commute.
\begin{corollary} \label{the:commutative-B}
In a P-algebra, if ``$\cdot$'' is commutative, then it is associative 
and the P-algebra is a Boolean algebra.
\end{corollary}
\begin{proof}
If $\cdot$ is commutative, {\bf Dot-monotonicity}, item~\ref{right-monotonicity} 
implies that $x \cdot y$ is a lower bound for $x$ and $y$.
Item~\ref{newzxy} in Theorem~\ref{the:basic1} then implies that $x \cdot y$ is 
a greatest lower bound for $x$ and $y$. 
We see that $X$ is a lattice under $\leq$ and that $\cdot$ is associative, since g.l.b.
is associative.
It is then easy to show that $P$ is a Boolean algebra.
\end{proof}

Our next result shows that the structure \mbox{$ \langle X , \, ' , \leq \rangle $} 
is an orthocomplemented poset: 
complementation is an antimonotone involution and $1$ is a top element.
It is a uniquely complemented poset as studied 
in~\cite{Waphare+1:2005,Chajda+2:2018}.
Item~\ref{S} prepares the proof of orthomodularity.
\begin{theorem} \label{the:orthocomplemented-poset}
In a P-algebra, for any \mbox{$ x , y \in X $} 
\begin{enumerate}
\item \label{involution}
\mbox{$ x'' = x $}.
\item \label{antitonicity}
\mbox{$ x \leq y $} iff \mbox{$ y' \leq x' $}.
\item \label{One-Right}
$1$ is a top element: \mbox{$ x \leq 1 $}.
\item \label{S}
if \mbox{$ x \smile y $}, then \mbox{$ x' \smile y $}.
\end{enumerate}
\end{theorem}
\begin{proof}{\ }
\begin{enumerate}
\item 
By Theorem~\ref{the:basic1}, item~\ref{z-commutation}, \mbox{$ x'' \cdot x' = $} 
\mbox{$ 0 \leq x $}.
By {\bf Dot-monotonicity}, item~\ref{right-monotonicity} we have
\mbox{$ x'' \cdot x \leq x $}.
By {\bf O}, we see that \mbox{$ x'' \leq x $}.

By {\bf Z}, \mbox{$ x \cdot x' = 0 \leq x'' $}.
By {\bf Dot-monotonicity}, item~\ref{right-monotonicity} we have
\mbox{$ x \cdot x'' \leq x'' $}.
By {\bf O}, we see that \mbox{$ x \leq x'' $}.
\item 
Assume \mbox{$ x \leq y $}. 
By item~\ref{involution} just above, \mbox{$ x \leq y'' $}.
By Theorem~\ref{the:basic1}, item~\ref{bot'}, \mbox{$ x \cdot y' = 0 $}. 
By Theorem~\ref{the:basic1}, item~\ref{z-commutation}, \mbox{$ y' \cdot x = 0 $}. 
By Theorem~\ref{the:basic1}, item~\ref{bot'}, \mbox{$ y' \leq x' $}.
The \emph{if} part follows by item~\ref{involution} just above.
\item 
By {\bf Z}, \mbox{$ 0 \leq x' $} and by item~\ref{antitonicity} just above
\mbox{$ x'' \leq 1 $}. Conclude by item~\ref{involution} above.
\item 
By assumption \mbox{$ x \cdot y \leq x $}.
By Theorem~\ref{the:basic1}, item~\ref{bot'} and 
item~\ref{involution} above, we have \mbox{$ ( x \cdot y ) \cdot x' = 0 $}.
By {\bf P-associativity}, item~\ref{xyz0}, we have 
\mbox{$ ( x' \cdot y ) \cdot x = 0 $}.
By Theorem~\ref{the:basic1}, item~\ref{bot'}, \mbox{$ x' \cdot y \leq x' $}, i.e.,
\mbox{$ x' \smile y $}.
\end{enumerate}
\end{proof}

\subsection{The operation $+$ and duality} \label{sec:+}
Boolean algebras are often, but not always, described as algebras with two binary 
operations related by the de Morgan laws. 
It is therefore only natural to define a dual operation to ``$\cdot$''.
\begin{definition} \label{def:+}
Let us define the operation $+$ by \mbox{$ x + y \eqdef ( y' \cdot x' )' $}. 
\end{definition}
Note that, in order that, in Section~\ref{sec:interpretation}, the action be focused close 
to the turnstile, the order of the operands has been reversed in the definition of the 
operation ``$+$''. 
The operation ``$+$'' does not seem to correspond to a logical 
notion that is usual or intuitive in general.
Note that in the plane, if $x$ and $y$ are generic lines, i.e., 
distinct and not orthogonal, then $x'$ and $y'$ are the respectively orthogonal lines,
\mbox{$ y' \cdot x' = x' $} and \mbox{$ x + y = x $}.
Note that $x + y$ is \emph{not} the subspace generated by $x$ and $y$.
This may be in relation with the fact that, when physicists consider superpositions of
states, those are typically superpositions of orthogonal states.
This section will study the operation ``$+$'' and show, in particular, that it behaves 
nicely when its arguments are orthogonal. 
In orthocomplemented lattices ``$+$'' is the l.u.b. operation.
In Hilbert spaces, when \mbox{$ A \bot B $}, \mbox{$ A + B $} is the subsace 
spanned by the union \mbox{$ A \cup B $}. 

\begin{theorem} \label{the:co+}
In a P-algebra, for any \mbox{$ x , y , z \in X $} 
\begin{enumerate}
\item \label{bottom-top}
\mbox{$ x + x' = x' + x = 1 $}.
\item \label{unique-complement}
if \mbox{$ x \cdot y = 0 $} and \mbox{$ x + y = 1 $}, then \mbox{$ y = x' $}.
\item \label{duality}
\mbox{$ x \cdot y = ( y' + x' )' $}.
\item \label{leq.+}
if \mbox{$ x \leq y $}, \mbox{$ y = $} \mbox{$ x + y = $} \mbox{$ y + x $}.
\item \label{+left-monotonicity}
\mbox{$ x \leq x + y $},
\item \label{+right-monotonicity}
if \mbox{$ x \leq y $}, then \mbox{$ z + x \leq z + y $},
\item \label{+ub}
if \mbox{$ x \leq z $} and \mbox{$ y \leq z $}, then \mbox{$ x + y \leq z $},
and therefore \mbox{$ 0 + 0 = 0 $}.
\item \label{+lub}
if \mbox{$ x \smile y $}, \mbox{$ x + y = l.u.b. ( x , y ) $}.
\end{enumerate}
\end{theorem}
\begin{proof}{\ }
\begin{enumerate}
\item 
Since, by {\bf Z} and Theorem~\ref{the:basic1}, item~\ref{z-commutation} 
\mbox{$ x \cdot x' = $} \mbox{$ x' \cdot x = 0 $} we have 
\mbox{$ ( x \cdot x' )' = $} \mbox{$ ( x' \cdot x )' = 1 $}, i.e.,
\mbox{$ x'' + x' = x' + x'' = 1 $}.
Conclude by item~\ref{involution} above.
\item 
Assume \mbox{$ x \cdot y = 0 $} and \mbox{$ x + y = 1 $}.
We have \mbox{$ ( y' \cdot x' )' = 0' $} and, by 
Theorem~\ref{the:orthocomplemented-poset}, item~\ref{involution},
\mbox{$ y' \cdot x' = 0 $}.
By Theorem~\ref{the:basic1}, item~\ref{bot'} and Theorem~\ref{the:co+}, 
item~\ref{involution} above we have 
\mbox{$ x \leq y' $} and \mbox{$ y' \leq x'' = x $}.
We see that \mbox{$ y' = x $} and \mbox{$ y = x' $}.
\item 
By Theorem~\ref{the:orthocomplemented-poset}, item~\ref{involution},
\mbox{$ y' + x' = $} \mbox{$ ( x'' \cdot y'' )' = $} \mbox{$ ( x \cdot y )' $} and
\mbox{$ ( y' + x' )' = $} \mbox{$ ( x \cdot y )'' = $} \mbox{$ x \cdot y $}.
\item 
By duality from Theorem~\ref{the:basic1}, item~\ref{small-large}, using 
Theorem~\ref{the:orthocomplemented-poset}, item~\ref{antitonicity}.
\item 
Dual of {\bf Dot-monotonicity}, item~\ref{right-monotonicity}.
\item 
Dual of {\bf Dot-monotonicity}, item~\ref{left-monotonicity}.
\item 
Dual of Theorem~\ref{the:basic1}, item~\ref{newzxy}.
\item 
Dual of Theorem~\ref{the:basic1}, item~\ref{smile-glb}.
\end{enumerate}
\end{proof}

The next theorem gathers different results.
Items~\ref{commute-+} and~\ref{half-distributivity} concern commuting features.
\begin{theorem} \label{the:+'}
For any \mbox{$ x , y , z \in X $}
\begin{enumerate}
\item \label{commute-+}
if \mbox{$ x \smile y $}, then \mbox{$ x + y = y + x $}.
\item \label{lub}
if \mbox{$ x \smile y $}, then \mbox{$ x + y $} is the l.u.b. of $x$ and $y$.
\item \label{half-distributivity}
if \mbox{$ x \smile y $}, 
then \mbox{$ x \cdot z + y \cdot z \leq $} \mbox{$ ( x + y ) \cdot z $}.
\item \label{cut-pre}
\mbox{$ x \leq x \cdot y + x \cdot y' $}.
\item \label{xxy}
\mbox{$ x \cdot ( x + y ) = x = ( x + y ) \cdot x $} and 
\mbox{$ x' \cdot ( x + y ) = $} \mbox{$ ( x + y ) \cdot x' $}. 
\item \label{pre-jump}
\mbox{$ x \leq y + z $} iff \mbox{$ x \cdot y' \leq z $}.
\end{enumerate}
\end{theorem}
\begin{proof}{\ }
\begin{enumerate}
\item 
By Theorem~\ref{the:orthocomplemented-poset}, item~\ref{S}, \mbox{$ x' \smile y' $} and, 
by Theorem~\ref{the:basic1}, item~\ref{smile-com}, we have \mbox{$ x' \cdot y' = $} 
\mbox{$ y' \cdot x' $} and, by Definition~\ref{def:basic}, \mbox{$ y + x = x + y $}.
\item 
By item~\ref{commute-+} just above and items~\ref{+left-monotonicity} 
and~\ref{+ub} in Theorem~\ref{the:co+}.
\item 
Assume \mbox{$ x \smile y $}. 
By item~\ref{commute-+} just above, \mbox{$ x + y = y + x $}.
By Theorem~\ref{the:co+}, item~\ref{+left-monotonicity} and {\bf Dot-monotonicity},
item~\ref{right-monotonicity}, we have \mbox{$ x \cdot z \leq $}
\mbox{$ ( x + y ) \cdot z $} and\break \mbox{$ y \cdot z \leq $} 
\mbox{$ ( y + x ) \cdot z = $} \mbox{$ ( x + y ) \cdot z $}.\\
By Theorem~\ref{the:co+}, item~\ref{+ub}, 
\mbox{$ x \cdot z + y \cdot z \leq ( x + y ) \cdot z $}.
\item 
By Theorem~\ref{the:basic1}, item~\ref{leq-bot}, 
\mbox{$ x \cdot y \, \bot \, x \cdot y' $}.
Therefore \mbox{$ x \cdot y \smile x \cdot y' $} and, by Theorem~\ref{the:+'},
item~\ref{commute-+}, \mbox{$ x \cdot y + x \cdot y' = $}
\mbox{$ x \cdot y' + x \cdot y $}.
By Theorem~\ref{the:co+}, item~\ref{+left-monotonicity}
\mbox{$ x \cdot y \leq x \cdot y + x \cdot y' $} and
\mbox{$ x \cdot y' \leq x \cdot y + x \cdot y' $}.
Conclude by {\bf O}.
\item 
By Theorem~\ref{the:co+}, item~\ref{+left-monotonicity} \mbox{$ x \leq x + y $}, 
i.e., \mbox{$ x \cdot ( x + y ) = x $}.
By Theorem~\ref{the:basic1}, item~\ref{leq-sim}, we have
\mbox{$ x \smile x + y $} and, by Theorem~\ref{the:basic1}, item~\ref{smile-com}
\mbox{$ x \cdot ( x + y ) = ( x + y ) \cdot x $}.
But, by Theorem~\ref{the:orthocomplemented-poset}, item~\ref{S}, we also have \mbox{$ x' \smile x + y $} 
and the last claim follows from Theorem~\ref{the:basic1}, item~\ref{smile-com}.
\item 
By Theorem~\ref{the:basic1}, items~\ref{bot'} and then~\ref{z-commutation}, 
{\bf P-associativity}, item~\ref{xyz0} and finally Theorem~\ref{the:basic1}, 
item~\ref{bot'} :
\[\begin{array}{c}
x \leq y + z {\rm \ iff \ } x \leq ( z' \cdot y' )' {\rm \ iff \ } x \cdot ( z' \cdot y' ) = 0 \\
{\rm \ iff \ } ( z' \cdot y' ) \cdot x = 0 {\rm \ iff \ } ( x \cdot y' ) \cdot z' = 0 
{\rm \ iff \ } x \cdot y' \leq z.
\end{array}\]
\end{enumerate}
\end{proof}

Our next result shows that the structure \mbox{$ \langle X , \, ' , \leq \rangle $} 
is an orthomodular poset.
\begin{theorem} \label{the:orthomodularity}
In a P-algebra, for any \mbox{$ x , y \in X $}, if \mbox{$ x \leq y $}, then
\begin{equation} \label{eq:orthomodular}
y = x + x' \cdot y = x + y \cdot x' = y \cdot x' + x = x' \cdot y + x.
\end{equation}
Note that, even though, in non-commutative structures, there may be many different 
notions of orthomodularity, all of them hold true in P-algebras.
\end{theorem}
\begin{proof}
By our assumption and Theorem~\ref{the:basic1}, item~\ref{small-large}, 
\mbox{$ x \cdot y = $} \mbox{$ x = $} \mbox{$ y \cdot x $}. 
By item~\ref{smile-com} there, \mbox{$ x \smile y $} and 
by Theorem~\ref{the:orthocomplemented-poset}, item~\ref{S}, \mbox{$ x' \smile y $} and 
\mbox{$ x' \cdot y = y \cdot x' $}.

By our assumption and Theorem~\ref{the:co+}, item~\ref{+left-monotonicity}, we have
\mbox{$ y \cdot x = $} \mbox{$ x \leq $} \mbox{$ x + y \cdot x' $} and
\mbox{$ y \cdot x' \leq $} \mbox{$ y \cdot x' + x $}.
Since \mbox{$ y \cdot x' \, \bot \, x $}, \mbox{$ y \cdot x' \smile x $} and, 
by Theorem~\ref{the:+'}, item~\ref{commute-+}, \mbox{$ y \cdot x' + x = $} 
\mbox{$ x + y \cdot x' $}.
We have \mbox{$ y \cdot x \leq x + y \cdot x' $} and 
\mbox{$ y \cdot x' \leq x + y \cdot x' $}. 
Condition {\bf O} in Definition~\ref{def:P-algebra} implies 
\mbox{$ y \leq x + y \cdot x' $}, and, by what we have seen at the start of this proof, 
\mbox{$ y \leq x + x' \cdot y $}.
By {\bf Dot-monotonicity}, item~\ref{right-monotonicity} and 
Theorem~\ref{the:co+}, item~\ref{+ub} we have \mbox{$ x + x' \cdot y \leq y $}.
We have shown the first equality.
The other equalities follow easily from what we have proven.
\end{proof}

\section{Language and interpretation} \label{sec:syntax}
We want to describe the logic of P-algebras: formulas, sequents~\cite{Gent:32, Gent:69} 
and semantics.
This section is fairly pedestrian: its only purpose is to prepare the ground 
for Section~\ref{sec:rules} and all the technical preparatory work has been done 
in Section~\ref{sec:P-algebras}.
The only point worth noticing is our interpretation of the sequents: 
association to the left on the left of the turnstile and to the right on its right.

\subsection{The language and its interpretation} \label{sec:language}
\subsubsection{Propositions and Sequents} \label{sec:propositions}
We consider a set $AT$ of atomic (caution: atomic here has nothing to do with its 
meaning in Appendix~\ref{sec:atomic}) propositions, a constant, one unary and 
one binary connective.
Most of the time, classical logic, the logic of Boolean algebras, is presented as a logic 
with two binary connectives: conjunction and disjunction.
Online with this paper's presentation of P-algebras, we include only one binary 
connective in the language, conjunction ($\wedge$).
Disjunction ($\vee$) is considered as a defined connective.

We shall represent propositions by small greek letters.
\begin{definition} \label{def:propositions}
\ 
\begin{itemize}
\item
An atomic proposition is a proposition.
\item
$\mathbf{0}$ is a proposition.
\item
If $\alpha$ is a proposition then \mbox{$ \neg \alpha$} is a proposition.
\item
If $\alpha $ and $\beta$ are propositions, then \mbox{$ \alpha \wedge \beta $} is a proposition.
\item
There are no other propositions.
\end{itemize}
\end{definition}
The set of propositions on $AT$ will be denoted by $\cal L$.

\begin{definition} \label{def:sequents}
A sequent is constituted by two finite sequences of propositions, 
separated by the turnstile symbol, that may be $\models$ or $\vdash$.
\end{definition} 
Here is a typical sequent:
\mbox{$ \alpha , \neg ( \neg \gamma \wedge \neg \neg \beta ) \vee \gamma \models \neg \neg ( \delta \wedge \epsilon ) $}.
In the representation of sequents we shall use greek capital letters 
to represent finite sequences of propositions.
The sequent \mbox{$ \Gamma \models \Gamma $}
is a sequent in which the same sequence $\Gamma$ of propositions appears 
on both sides of the turnstile.

\subsubsection{Interpretation} \label{sec:interpretation}
Propositions are interpreted as features in a P-algebra.
In a P-algebra \mbox{$ \langle X , 0 , ' , \cdot \rangle $}, 
a proposition is interpreted as a member of $X$ with the help of an assignment 
function for atomic propositions.
\begin{definition} \label{def:interpretation}
Assume a P-algebra \mbox{$ \langle X , 0 , ' , \cdot \rangle $} and 
an assignment \mbox{$ v : AT \longrightarrow X $}.
The assignment $v$ can be extended to the language $\cal L$, 
\mbox{$ v : {\cal L} \longrightarrow X $} by
\begin{itemize}
\item
\mbox{$ v(\mathbf{0}) = 0 $},
\item
\mbox{$ v( \neg \alpha ) = v(\alpha)' $} for any \mbox{$\alpha \in {\cal L}$},
\item
\mbox{$ v( \alpha \wedge \beta ) = v(\alpha) \cdot v(\beta) $} for any 
\mbox{$\alpha , \beta \in {\cal L}$}.
\end{itemize}
\end{definition}
\begin{definition} \label{def:equivalence}
For any \mbox{$ \alpha , \beta \in {\cal L} $} we shall say that $\alpha$ and $\beta$
are {\em semantically equivalent} and write \mbox{$ \alpha \equiv \beta $} iff 
\mbox{$ v(\alpha) = v(\beta) $} for any P-algebra and any assignment of features 
to the atomic propositions.
\end{definition}

Our first result asserts that we can eliminate double negations.
\begin{lemma} \label{the:dnegation}
For any \mbox{$ \alpha , \beta \in {\cal L} $}, one has 
\mbox{$ \neg ( \neg \alpha ) \equiv \alpha $}.
\end{lemma}
\begin{proof}
\mbox{$ v ( \neg ( \neg \alpha ) ) = $} \mbox{$ v( \alpha)'' =$} \mbox{$ v( \alpha) $} 
by item~\ref{involution} in Theorem~\ref{the:orthocomplemented-poset}.
\end{proof}

\begin{definition} \label{def:disjunction}
Define the disjunction of two propositions \mbox{$ \alpha , \beta \in {\cal L} $} by  
\begin{equation} \label{eq:disjunction}
\alpha \vee \beta \eqdef \neg ( \neg \beta \wedge \neg \alpha ).
\end{equation}
\end{definition}

\begin{lemma} \label{the:disjunction}
For any \mbox{$ \alpha , \beta \in {\cal L} $}, \mbox{$ v ( \alpha \vee \beta ) = $}
\mbox{$ v ( \alpha ) + v ( \beta ) $} and 
\mbox{$ \alpha \wedge \beta \equiv $} 
\mbox{$ \neg ( \neg \beta \vee \neg \alpha ) $}.
\end{lemma}
\begin{proof}
The first claim follows from Definitions~\ref{def:disjunction} and~\ref{def:+}.
By Definition~\ref{def:disjunction} and Lemma \ref{the:dnegation}, 
\mbox{$ \neg ( \neg \beta \vee \neg \alpha ) = $}
\mbox{$ \neg ( \neg ( \neg \neg \alpha \wedge \neg \neg \beta ) \equiv $}
\mbox{$ \alpha \wedge \beta $}.
\end{proof}

We must now interpret sequents. 
We shall, as expected, interpret the left-hand side and the right-hand side as 
features and the turnstile $\models$ as implication ($\leq$).
In the literature, following Gentzen~\cite{Gent:32}, the comma on the left-hand 
side is interpreted as a sort of conjunction ($\cdot$) and the
comma on the right-hand side as sort of disjunction ($+$).
Since those operations are not associative, we must decide how to associate the 
elements of the left-hand side and how to associate those of the right-hand side.
Since \mbox{$\alpha \wedge \beta$} denotes the result of measuring $\beta$ 
\emph{after} $\alpha$, it is natural to decide that the elements of the left-hand side 
associate to the left.
To keep the action close to the turnstile, we decide that the elements of the 
right-hand side associate to the right.

As a consequence, the interpretation, in a P-algebra 
\mbox{$\langle X , 0 , ' , \cdot \rangle $}, for assignment 
\mbox{$ v: {\cal L} \rightarrow X $}, of sequent:
\begin{equation} \label{eq:sequent}
\alpha_{0} , \alpha_{1} \ldots \alpha_{n - 1} \models_{v} \beta_{0} , \beta_{1}  \ldots 
\beta_{m - 1}
\end{equation}
is 
\begin{equation} \label{eq:seq-interpretation}
v ( ( \ldots ( \alpha_{0} \wedge \alpha_{1} ) \wedge \ldots ) \wedge \alpha_{n - 1} ) 
\leq v ( \beta_{0} \vee ( \beta_{1} \vee ( \dots \vee \beta_{m - 1}) \ldots ) ).
\end{equation}
If the right-hand side of the turnstile is empty its interpretation is $0$.
If the left-hand side of the turnstile is empty its interpretation is $1$.

\begin{definition} \label{def:valid-sequent}
A sequent is \emph{valid in a P-algebra \mbox{$ \langle X , 0 , ' , \cdot \rangle $} }
iff its interpretation holds for every assignment.
It is \emph{valid} iff it is valid in any P-algebra.
\end{definition}

For example, the sequent \mbox{$ \alpha \models \alpha , \alpha $} is valid 
for any $\alpha$ since \mbox{$ x = x + x $} for any feature $x$.
Our next lemma shows that propositions can jump over the turnstile in both direction 
at the cost of an added negation. 
\begin{theorem} \label{the:jump}
For any P-algebra \mbox{$ P = \langle X , 0 , ' , \cdot \rangle $}, 
\mbox{$ \alpha_{0} , \ldots , \alpha_{n} \models \beta_{0} , \beta_{1} , \ldots , \beta_{m} $} 
is valid in $P$ iff 
\mbox{$  \alpha_{0} , \ldots , \alpha_{n} , \neg \beta_{0} \models 
\beta_{1} , \ldots , \beta_{m} $} is valid in $P$.
Also,
\mbox{$ \alpha_{0} , \ldots , \alpha_{n} \models \beta_{0} ,\break  \beta_{1} , \ldots , \beta_{m} $}
is valid in $P$ iff 
\mbox{$  \alpha_{0} , \ldots , \alpha_{n - 1} \models \neg \alpha_{n} , \beta_{0} ,
\beta_{1} , \ldots , \beta_{m} $} is valid in $P$.
\end{theorem}
\begin{proof}
Theorem~\ref{the:+'}, item~\ref{pre-jump} asserts the first claim.
The second claim follows by removing a double negation.
\end{proof}
The nature of the quantum negation has been discussed in the literature, e.g., 
in~\cite{Connectives:Fest}.
On one hand since orthogonal ($\bot$) is stronger than distinct ($\neq$) one may be 
tempted to conclude that quantum negation ($\neg$) is a strong negation, 
possibly akin to an intuitionistic negation. 
Theorem~\ref{the:jump} above shows that this is not so: 
\emph{quantum negation is classical}.

In view of Lemma~\ref{the:dnegation} and Theorem~\ref{the:jump}, 
we can restrict our attention to sequents with an empty right-hand side, 
sequents of the form \mbox{$ \alpha_{0} , \ldots , \alpha_{n} \models $}. 
The intuitive meaning of such a sequent is:
the sequence of measurements $\alpha_{0}$, $\alpha_{1}$, \ldots $\alpha_{n}$ 
\emph{in this order} will never be observed.
The following will help in the study of the validity of one-sided sequents.
\begin{definition} \label{def:phi-psi}
For any sequence of propositions: 
\mbox{$ \Sigma = $} \mbox{$ \sigma_{0} , \sigma_{1} , \ldots , \sigma_{n} $}
let us denote by \mbox{$ \phi(\Sigma) $} the proposition 
\mbox{$ ( \ldots ( \sigma_{0} \wedge \sigma_{1} ) \wedge \ldots ) \wedge \sigma_{n} $}, i.e., the left-associated conjunction of the propositions in the sequence.
Note that, for any sequences $\Gamma$ and $\Delta$, 
\mbox{$ \phi ( \Gamma , \Delta ) = $} \mbox{$ \phi ( \phi ( \Gamma ) , \Delta ) $}.
Similarly \mbox{$ \psi( \Sigma ) $} will denote the proposition
\mbox{$ \neg \sigma_{0} \vee ( \neg \sigma_{1} \vee ( \ldots \neg \sigma_{n} ) ) $}, 
i.e., the right-associated disjunction of the negations of the propositions 
in the sequence.
Note that, for any sequences $\Gamma$ and $\Delta$, 
\mbox{$ \psi ( \Gamma , \Delta ) = $} \mbox{$ \psi ( \Gamma , \psi ( \Delta ) ) $}.
\end{definition}


Equation~(\ref{eq:seq-interpretation}) and Theorem~\ref{the:jump} imply:
\begin{equation} \label{eq:Gamma-Delta}
{\rm a \ sequent \ } \Gamma , \Delta \models  {\rm \ is \ valid \ iff \ }
\phi( \Gamma ) \leq \psi( \Delta ).
\end{equation}

\section{The deductive system $\cal R$} \label{sec:rules}
We shall describe a system $\cal R$ of eight deduction rules 
and prove it is sound and complete for the logic of P-algebras.
A table of these inference rules can be found in figure~\ref{fig:rules}.

In Section~\ref{sec:soundness} the rules are described and proved to be valid:
a classical Cut rule, an Exchange rule limited to a sequence of three propositions, 
two limited Weakening rules, an introduction rule for the constant {\bf 0}, an 
introduction rule for negation and three $\wedge$ introduction-elimination rules.
There are many equivalent systems and $\cal R$ may not be the system with 
optimal proof-theoretic properties.
Section~\ref{sec:derived} proves the validity of a number of derived rules.
Section~\ref{sec:implication} provides an in-depth study of the properties of the
deductive system.

\subsection{A sound deductive system} \label{sec:soundness}
In deduction rules we use the symbol $\vdash$ to separate the left side 
from the right side of a sequent and not $\models$ as above.
A deduction rule consists of a finite set of sequents, \emph{the assumptions} 
and a sequent, \emph{the conclusion} separated by a horizontal line, 
called the inference line.
A double horizontal line signals a bi-directional rule: it can be used in both directions.
For example, consider the following Cut rule. 

\[ \begin{array} {lc} \\
{\bf Cut} &
\begin{array}{c}
\Gamma , \alpha , \Delta \vdash \ \ \ \ \Gamma , \neg \alpha , \Delta \vdash \\
\hline
\Gamma , \Delta \vdash
\end{array} 
\end{array} \]
Such a rule is meant to be part of a set of deduction rules and its meaning is: 
if one has already established the two sequents above the inference line, then, one is
entitled to establish the sequent below the inference line.
We are interested in two properties of such rules and systems of rules.
\begin{definition} \label{def:sound-complete}
An deduction rule is said to be \emph{sound} iff, for any P-algebra and 
any assignment $v$, for which all the assumptions are valid, 
the conclusions are also valid (in the specified P-algebra, 
with the specified assignment $v$).
A set of deduction rules is \emph{complete} iff any sequent that is \emph{valid} 
(in all P-algebras, for all assignments) can be derived using only the rules in the set.
\end{definition}

Let's check that {\bf Cut} is indeed sound.
In view of Equation~(\ref{eq:Gamma-Delta}), to show the soundness of {\bf Cut},
it is enough to show that 
\mbox{$ v( \phi( \Gamma ) ) \cdot v (\alpha ) \leq v( \psi( \Delta ) ) $} and 
\mbox{$ v( \phi( \Gamma ) ) \cdot v( \alpha)' \leq $} \mbox{$ v( \psi( \Delta ) ) $} 
imply \mbox{$ v( \phi( \Gamma ) ) \leq v( \psi( \Delta ) ) $}.
This is guaranteed by property {\bf O} of Definition~\ref{def:P-algebra}.

Let us now consider structural rules, i.e., rules that do not involve the connectives.
There is no valid general {\bf Exchange} rule: one cannot modify the order of the 
propositions in the left-hand side of a sequent, but there is a very limited exchange 
rule: if a sequent has only three propositions, the order of these three propositions 
may be reversed.
This is the counterpart to {\bf P-associativity}, item~\ref{xyz0} and its soundness
follows from it.
\[ \begin{array} {lc} \\
{\bf Circ} &
\begin{array}{c}
\alpha , \beta , \gamma \vdash \\
\hline
\gamma , \beta , \alpha \vdash
\end{array} 
\end{array} \]

Concerning Weakening the situation is more complex: one cannot add a 
proposition anywhere in a sequent. 
A first Weakening rule allows the introduction of a proposition at the extremities 
of a sequent, on the left and on the right. 
\[ \begin{array} {lc} \\
{\bf EWeakening} &
\begin{array}{c}
\Sigma \vdash \\
\hline
\Gamma , \Sigma , \Delta \vdash
\end{array} 
\end{array} \]
Let us show that {\bf EWeakening} is valid.
By assumption \mbox{$ v ( \phi ( \Sigma ) ) = 0 $}.
By Theorem~\ref{the:basic1}, item~\ref{x0}, 
\mbox{$ v ( \phi ( \Gamma ) \cdot v ( \phi (  \Sigma ) )  = 0 $}.
By a repeated use of {\bf P-associativity}, item~\ref{xyz0}, one can show, by induction
on the size of $\Sigma$, that 
\mbox{$ v ( \phi ( \Gamma , \Sigma ) ) = $}
\mbox{$ v ( \phi ( \Gamma ) \wedge \phi ( \Sigma ) ) = $}
\mbox{$ v ( \phi ( \Gamma ) ) \cdot v ( \phi (  \Sigma ) )  = 0 $}.
Now, let \mbox{$ \Delta = \delta_{0} , \ldots , \delta_{n - 1} $}.
We have, by {\bf Z},
\[
v ( \phi ( \Gamma , \Sigma , \Delta ) ) = 
v ( \ldots ( \phi ( \Gamma , \Sigma ) \wedge \delta_{0} ) \ldots \delta_{n-1} ) = 0.
\]
 
A second Weakening rule allows the introduction of a proposition 
in the midst of a sequent: but only if the proposition is guaranteed to hold at this 
position, i.e., if its negation cannot hold at this point.
It allows some sort of \emph{Stuttering}.
\[ \begin{array} {lc} \\
{\bf MWeakening} &
\begin{array}{c}
\Gamma , \Delta \vdash \ \ \ \ \Gamma , \neg \alpha \vdash \\
\hline
\Gamma , \alpha , \Delta \vdash 
\end{array} 
\end{array} \]
The {\bf MWeakening} rule is sound since, 
if \mbox{$ \Gamma , \neg \alpha \vdash $},
\mbox{$ v( \phi( \Gamma ) ) \leq v( \alpha ) $} and 
\[
v( \phi( \Gamma , \alpha ) ) = v( \phi( \Gamma ) ) \cdot v( \alpha ) = 
v( \phi( \Gamma ) ).
\]
Therefore 
\[
v ( \phi ( \Gamma , \alpha , \Delta ) ) = 
v ( \phi ( \phi ( \Gamma , \alpha ) , \Delta ) ) =
v ( \phi ( \phi ( \Gamma ) , \Delta ) ) = 
v ( \phi ( \Gamma , \Delta ) ).
\]

An introduction rule for the individual constant {\bf 0}.
It is an axiom, i.e., a deduction rule with no assumptions.
\[ \begin{array}{lc} \\
{\bf 0Axiom} &
\begin{array}{c}
\\
\hline
{\bf 0} \vdash
\end{array}
\end{array} \]
Soundness follows from \mbox{$ v( {\bf 0} ) = 0 $} and {\bf Z}.

An introduction rule for negation.
It is an axiom.
\[ \begin{array} {lc} \\
{\bf NAxiom} &
\begin{array}{c}
\\
\hline
\alpha , \neg \alpha \vdash
\end{array} 
\end{array} \]
It is sound by {\bf Comp}.
We could have restricted the rule to atomic propositions and derived the full rule
in Section~\ref{sec:derived} and there may be advantages in studying such a 
seemingly weaker system in future work.

We have two introduction-elimination rules for $\wedge$.
They are bi-directional rules, denoted by a double line.
One may deduce the sequents below the double line from the ones above the double
line, but one may also deduce the ones above the line from the ones below the line. 

The first one introduces or eliminates a conjunction in the leftmost part of a sequent. 
The rule that introduces or eliminates a conjunction in the rightmost part of a sequent, 
i.e., close to the turnstile can be derived from it, see {\bf LR-$\wedge$} in 
Section~\ref{sec:wedge-intro}.
\[\begin{array}{lc} \\
{\bf LL-\wedge} &
\begin{array}{c}
\alpha , \beta , \Delta \vdash \\
\hline
\hline
\alpha \wedge \beta , \Delta \vdash
\end{array} 
\end{array} \]
The soundness of both directions in {\bf LL$\wedge$} follows from 
Equation~(\ref{eq:seq-interpretation}) and the fact that 
\mbox{$ \phi( \alpha , \beta , \Delta ) = $} 
\mbox{$ \phi( \alpha \wedge \beta , \Delta ) $}.

The second $\wedge$ introduction-elimination rule allows the conjunction (on the left) 
of a proposition that is guaranteed to hold at this point.
\[ \begin{array} {lc} \\
{\bf ML-\wedge} &
\begin{array}{c}
\Gamma , \beta , \Delta \vdash \ \ \ \ \Gamma , \neg \alpha \vdash \\
\hline
\hline
\Gamma , \alpha \wedge \beta , \Delta \vdash \ \ \ \ \Gamma , \neg \alpha \vdash
\end{array} 
\end{array} \]
The soundness of {\bf ML-$\wedge$} follows from the fact that, if 
\mbox{$ v( \phi ( \Gamma ) ) \leq v ( \alpha ) $}, then 
\[
v ( \phi ( \Gamma ) ) \cdot ( v ( \alpha ) \cdot v ( \beta ) ) = 
( v ( \phi ( \Gamma ) ) \cdot v ( \alpha ) ) \cdot v ( \beta ) = 
v ( \phi ( \Gamma ) ) \cdot \ v ( \beta ) 
\] 
by {\bf P-associativity}, item~\ref{large-small-left}.

\begin{figure} 
\[ \begin{array} {rl} 
{\bf Cut} &
\begin{array}{c}
\Gamma , \alpha , \Delta \vdash \ \ \ \ 
\Gamma , \neg \alpha , \Delta \vdash \\
\hline
\Gamma , \Delta \vdash
\end{array} 
\end{array} \]
\[ \begin{array}{rl}
{\bf Circ} & 
\begin{array}{c}
\alpha , \beta , \gamma \vdash \\
\hline
\gamma , \beta , \alpha \vdash
\end{array} 
\end{array} \]
\[ \begin{array} {rlrl} 
{\bf EWeakening} & 
\begin{array}{c}
\Sigma \vdash \\
\hline
\Gamma , \Sigma , \Delta \vdash 
\end{array} 
& {\bf MWeakening} &
\begin{array}{c}
\Gamma , \Delta \vdash \ \ \ \  \Gamma , \neg \alpha \vdash \\
\hline
\Gamma , \alpha , \Delta \vdash 
\end{array}
\end{array} \]
\[ \begin{array}{rlrl}
{\bf 0Axiom} &
\begin{array}{c}
\\
\hline
{\bf 0} \vdash 
\end{array}
& {\bf NAxiom} &
\begin{array}{c}
\\
\hline
\alpha , \neg \alpha \vdash
\end{array} 
\end{array} \]
\[ \begin{array}{rlrl} 
{\bf LL-\wedge} &
\begin{array}{c}
\alpha , \beta , \Delta \vdash \\
\hline
\hline
\alpha \wedge \beta , \Delta \vdash
\end{array} 
& {\bf ML-\wedge} 
\begin{array}{c}
\Gamma , \beta , \Delta \vdash \ \ \Gamma , \neg \alpha \vdash \\
\hline 
\hline
\Gamma , \alpha \wedge \beta , \Delta \vdash \ \ \Gamma , \neg \alpha \vdash
\end{array}
\end{array} \]
\caption{Deductive system {\cal R}} \label{fig:rules}
\end{figure}

One easily sees that if the standard unlimited Exchange rule is added to our 
system, one obtains classical propositional logic where $\neg$ is negation, 
$\wedge$ is conjunction and $\vee$ is disjunction. 
We can now state a soundness theorem: its proof has been provided above.
\begin{theorem} \label{the:soundness}
Each of the eight deductive rules of the system $\cal R$ described above is sound.
\end{theorem}

\subsection{Derived rules} \label{sec:derived}
To prepare the completeness result of Theorem~\ref{the:completeness} we need
to put in evidence the power of the deductive system presented above and study 
the logic that the system embodies.
Most of the work will be done in Section~\ref{sec:implication}, 
but a number of basic results will be proved first.
By Theorem~\ref{the:soundness}, the rules presented below are sound, but 
our purpose is to show more: they can be derived in the system $\cal R$.
Figure~\ref{fig:derived} presents a table of those derived rules.

\subsubsection{Repetition and contraction} \label{sec:repetition}
\begin{itemize}
\item
A repetition rule:
\[ \begin{array} {rl} 
{\bf Repetition} &
\begin{array}{c}
\Gamma , \alpha , \Delta \vdash \\
\hline
\Gamma , \alpha , \alpha , \Delta \vdash
\end{array} 
\end{array} \]
Derivation:
\[
\begin{array}{rlrl}
{\bf NAxiom} & \alpha , \neg \alpha \vdash & & \\
\cline{2-2}
{\bf EWeakening} & \Gamma , \alpha , \neg \alpha \vdash 
& {\bf Assumption} & \Gamma , \alpha , \Delta \vdash \\
\cline{2-4}
& {\bf MWeakening} & \Gamma , \alpha , \alpha , \Delta \vdash
\end{array}
\]
\item
A contraction rule:
\[ \begin{array} {rl} \\
{\bf Contraction} &
\begin{array}{c}
\Gamma , \alpha , \alpha , \Delta \vdash \\
\hline
\Gamma , \alpha , \Delta \vdash
\end{array} 
\end{array} \]
Derivation:
\[
\begin{array}{rlrl}
{\bf NAxiom} & \alpha , \neg \alpha \vdash \\
\cline{2-2}
{\bf EWeakening} & \Gamma , \alpha , \neg \alpha , \Delta \vdash  
& {\bf Assumption} & \Gamma , \alpha , \alpha , \Delta \vdash \\
\cline{2-4}
& {\bf Cut} & \Gamma , \alpha , \Delta \vdash
\end{array}
\]
\end{itemize}

\subsubsection{An exchange rule} \label{sec:exchange}
In a sequent of two propositions, they can be exchanged.
\[ \begin{array}{rl}
{\bf Exchange} & 
\begin{array}{c}
\alpha , \beta \vdash \\
\hline
\beta , \alpha \vdash
\end{array}
\end{array} \]
Derivation:
\[ \begin{array}{rl}
{\bf Assumption} & \alpha , \beta \vdash \\
\cline{2-2}
{\bf Repetition} & \alpha , \beta , \beta \vdash \\
\cline{2-2}
{\bf Circ} & \beta , \beta , \alpha \vdash \\
\cline{2-2}
{\bf Contraction} & \beta , \alpha \vdash
\end{array} \]

\subsubsection{Double negations} \label{sec:double-neg}
Double negations can be eliminated and also introduced.
\[ \begin{array}{rl}
{\bf D1} &
\begin{array}{c}
\Gamma , \neg \neg \alpha , \Delta \vdash \\
\hline
\Gamma , \alpha , \Delta \vdash  
\end{array}
\end{array} \]
Derivation:
\[\footnotesize \begin{array}{rlrlrl}
& & {\bf NAxiom} & \neg \alpha , \neg \neg \alpha \vdash \\
\cline{4-4}
& & {\bf Exchange} & \neg \neg \alpha , \neg \alpha \vdash 
& {\bf NAxiom} & \alpha , \neg \alpha \vdash \\
\cline{4-4} \cline{6-6}
{\bf Assumption} & \Gamma , \neg \neg \alpha , \Delta \vdash 
& {\bf EWeakening} & \Gamma , \neg \neg \alpha , \neg \alpha \vdash
& {\bf Exchange} & \neg \alpha , \alpha \vdash \\
\cline{2-4} \cline{6-6}
& {\bf MWeakening} & \Gamma , \neg \neg \alpha , \alpha , \Delta \vdash
& & {\bf EWeakening} & \Gamma , \neg \alpha , \alpha , \Delta \vdash \\
\cline{3-6}
& & & {\bf Cut} & \Gamma , \alpha , \Delta \vdash
\end{array} \]

\[ \begin{array}{rl}
{\bf D2} & 
\begin{array}{c}
\Gamma , \alpha , \Delta \vdash \\
\hline
\Gamma , \neg \neg \alpha , \Delta \vdash  
\end{array}
\end{array} \]
Derivation:
\[ \begin{array}{rlrl}
{\bf NAxiom} &  \neg \neg \alpha , \neg \neg \neg \alpha \vdash \\
\cline{2-2}
{\bf EWeakening} & \Gamma , \neg \neg \alpha , \neg \neg \neg \alpha \vdash 
& {\bf Assumption} & \Gamma , \alpha , \Delta \vdash \\
\cline{2-2} \cline{4-4}
{\bf D1} & \Gamma , \alpha , \neg \neg \neg \alpha \vdash
& {\bf NAxiom} & \neg \alpha , \neg \neg \alpha \vdash \\
\cline{2-2} \cline{4-4}
{\bf MWeakening} & \Gamma , \alpha , \neg \neg \alpha , \Delta \vdash 
& {\bf EWeakening} & \Gamma , \neg \alpha , \neg \neg \alpha , \Delta \vdash  \\
\cline{2-4}
& {\bf Cut} & \Gamma , \neg \neg \alpha , \Delta \vdash
\end{array} \]

\subsubsection{Equivalence rule} \label{sec:equivalence-rule}
Our last derived rule describes a condition that implies that a proposition can replace 
another one in any context.

\[ \begin{array} {lc} \\
{\bf Equiv} &
\begin{array}{c}
\Gamma , \alpha , \Delta \vdash \ \ \alpha , \neg \beta \vdash \ \ \beta , \neg \alpha \vdash \\
\hline
\Gamma , \beta , \Delta \vdash
\end{array} 
\end{array} \]
Derivation:
\[\footnotesize \begin{array}{lclclc}
{\bf Assumption} & \alpha , \neg \beta \vdash &
& & {\bf Assumption} &  \beta , \neg \alpha \vdash \\
\cline{2-2} \cline{6-6}
{\bf EWeakening} & \Gamma , \alpha , \neg \beta \vdash 
& {\bf Assumption} & \Gamma , \alpha , \Delta \vdash 
& {\bf Exchange} &  \neg \alpha , \beta \vdash \\
\cline{2-4} \cline{6-6}
& {\bf MWeakening} & \Gamma , \alpha , \beta , \Delta \vdash & 
& {\bf EWeakening} & \Gamma , \neg \alpha , \beta , \Delta \vdash \\
\cline{3-6}
& & {\bf Cut} & \Gamma , \beta , \Delta \vdash
\end{array} \]

\subsubsection{A $\wedge$ introduction and elimination rule} \label{sec:wedge-intro}
\[ \begin{array}{rl}
{ \bf LR-\wedge}  &
\begin{array}{c}
\Gamma , \alpha , \beta \vdash \\
\hline
\hline
\Gamma , \beta \wedge \alpha \vdash 
\end{array} 
\end{array} \]
Note the change of order between $\alpha$ and $\beta$ in {\bf LR-$\wedge$}.
For the derivation from top to bottom:
\[ \begin{array}{rl}
{\bf Assumption} & \Gamma , \alpha , \beta \vdash \\
\cline{2-2}
{\bf LL-\wedge} & \phi ( \Gamma ) , \alpha , \beta \vdash \\
\cline{2-2}
{\bf Circ} & \beta , \alpha , \phi ( \Gamma ) \vdash \\
\cline{2-2}
{\bf LL-\wedge} & \beta \wedge \alpha , \phi ( \Gamma ) \vdash \\
\cline{2-2}
{\bf Exchange} & \phi ( \Gamma ) , \beta \wedge \alpha \vdash \\
\cline{2-2}
{\bf LL-\wedge} & \Gamma , \beta \wedge \alpha \vdash
\end{array} \]

For the bottom to top direction.
\[ \begin{array}{rl}
{\bf Assumption} & \Gamma , \beta \wedge \alpha \vdash \\
\cline{2-2}
{\bf LL-\wedge} & \phi ( \Gamma ) , \beta \wedge \alpha \vdash \\
\cline{2-2}
{\bf Exchange} & \beta \wedge \alpha, \phi ( \Gamma ) \vdash \\
\cline{2-2}
{\bf LL-\wedge} & \beta , \alpha , \phi ( \Gamma ) \vdash \\
\cline{2-2}
{\bf Circ} & \phi ( \Gamma ) , \alpha , \beta \vdash \\
\cline{2-2}
{\bf \wedge LL-\wedge} & \Gamma , \alpha , \beta \vdash
\end{array} \]

\begin{figure}
\[ \begin{array}{rlrl}
{\bf Repetition} &
\begin{array}{c}
\Gamma , \alpha , \Delta \vdash \\
\hline
\Gamma , \alpha , \alpha , \Delta \vdash
\end{array} 
& {\bf Contraction} &
\begin{array}{c}
\Gamma , \alpha , \alpha , \Delta \vdash \\
\hline
\Gamma , \alpha , \Delta \vdash
\end{array} 
\end{array} \]
\[ \begin{array}{lr}
{\bf D} &
\begin{array}{c}
\Gamma , \neg \neg \alpha , \Delta \vdash \\
\hline
\hline
\Gamma , \alpha , \Delta \vdash  
\end{array} 
\end{array} \]
\[ \begin{array}{lrlr}
{\bf Exchange} & 
\begin{array}{c}
\alpha , \beta \vdash \\
\hline
\beta , \alpha \vdash
\end{array} 
& { \bf LR-\wedge}  &
\begin{array}{c}
\Gamma , \alpha , \beta \vdash \\
\hline
\Gamma , \beta \wedge \alpha \vdash \\
\end{array} 
\end{array} \]
\[ \begin{array}{rl}
{\bf Equiv} &
\begin{array}{c}
\Gamma , \alpha , \Delta \vdash \ \ \alpha , \neg \beta \vdash \ \ \beta , \neg \alpha \vdash \\
\hline
\Gamma , \beta , \Delta \vdash
\end{array}
\end{array} \]
\caption{Derived rules} \label{fig:derived}
\end{figure}

\subsection{Implication and logical equivalence} \label{sec:implication}
We shall begin, now, an in-depth study of the deductive system $\cal R$. 
\begin{definition} \label{def:implication}
Let \mbox{$ \alpha , \beta \in {\cal L} $}. We shall say that $\beta$ \emph{implies} $\alpha$ 
and write \mbox{$ \beta \rightarrow \alpha $}
iff the sequent \mbox{$ \beta , \neg \alpha \vdash $} is derivable 
in the system $\cal R$.
We shall say that $\alpha$ and $\beta$ are \emph{logically equivalent} and write 
\mbox{$ \alpha \simeq \beta $} iff \mbox{$ \alpha \rightarrow \beta $} and 
\mbox{$ \beta \rightarrow \alpha $}.
\end{definition}
Our first task is to characterize the relations $\rightarrow$ and $\simeq$.
Lemma~\ref{the:implication} shows that implication is a preorder.
\begin{lemma} \label{the:implication}
The implication relation $\rightarrow$ is reflexive and transitive.
Therefore the relation $\simeq$ is an equivalence relation.
\end{lemma}
\begin{proof}
Reflexivity follows from {\bf NAxiom}.
For transitivity, assume \mbox{$ \alpha \rightarrow \beta $} and\break 
\mbox{$ \beta \rightarrow \gamma $}.
The following derivation shows that \mbox{$ \alpha \rightarrow \gamma $}.

\[ \begin{array}{rlrl}
{\bf Assumption} & \alpha , \neg \beta \vdash 
& {\bf Assumption} & \beta , \neg \gamma \vdash \\
\cline{2-2} \cline{4-4} 
{\bf EWeakening} & \alpha , \neg \beta , \neg \gamma \vdash & 
{\bf EWeakening} & \alpha , \beta , \neg \gamma \vdash \\
\cline{2-4}
& {\bf Cut} & \alpha , \neg \gamma \vdash
\end{array} \]
\end{proof}

\begin{lemma} \label{the:equivalence}
For any \mbox{$ \alpha , \beta \in {\cal L} $}, \mbox{$ \alpha \rightarrow \beta $}
iff for any sequence $\Delta$, we have 
\mbox{$ \beta , \Delta \vdash $} implies
\mbox{$ \alpha , \Delta \vdash $}.
Also \mbox{$ \alpha \simeq \beta $} iff for any sequences $\Gamma$ and $\Delta$, one has \mbox{$ \Gamma , \alpha , \Delta \vdash $} iff \mbox{$ \Gamma , \beta , \Delta \vdash $}.
\end{lemma}
\begin{proof}
Let \mbox{$ \alpha \rightarrow \beta $}.
Consider the following derivation:
\[ \begin{array}{rlrl}
{\bf Assumption} & \alpha , \neg \beta \vdash
& {\bf Assumption} & \beta , \Delta \vdash \\
\cline{2-2} \cline{4-4}
{\bf EWeakening} & \alpha , \neg \beta , \Delta \vdash 
& {\bf EWeakening} & \alpha , \beta , \Delta \vdash \\
\cline{2-4}
& {\bf Cut} & \alpha , \Delta \vdash
\end{array} \]
We have shown the \emph{only if} part of the first claim.
For the \emph{if} part, assume that \mbox{$ \beta , \Delta \vdash $} implies
\mbox{$ \alpha , \Delta \vdash $}.
By {\bf NAxiom} we have \mbox{$ \beta , \neg \beta \vdash $} and therefore
\mbox{$ \alpha , \neg \beta \vdash $}.

For the second claim, the \emph{only if} part is proved by the derived rule 
{\bf Equiv}.
The \emph{if} part follows from what we have just proved about implication.
\end{proof}

Our next goal is to show that $\simeq$ is a congruence relation for the three 
connectives.
When possible, we put in evidence the properties of the relation $\rightarrow$.
\begin{lemma} \label{the:congruence}
For any \mbox{$ \alpha , \beta , \gamma \in {\cal L} $}:
\begin{enumerate}
\item \label{contradiction}
\mbox{$ \alpha \vdash $} iff \mbox{$ \alpha \simeq {\bf 0} $}.
\item \label{negation}
if \mbox{$ \alpha \rightarrow \beta $}, then
\mbox{$ \neg \beta \rightarrow \neg \alpha $} and therefore if 
\mbox{$ \alpha \simeq \beta $}, then \mbox{$ \neg \alpha \simeq \neg \beta $}.
\item \label{cdot-left}
if \mbox{$ \alpha \rightarrow \beta $}, then 
\mbox{$ \alpha \wedge \gamma \rightarrow \beta \wedge \gamma $}
and therefore if \mbox{$ \alpha \simeq \beta $}, then 
\mbox{$ \alpha \wedge \gamma \simeq \beta \wedge \gamma $}.
\item \label{cdot-right}
if \mbox{$ \alpha \simeq \beta $}, then 
\mbox{$ \gamma \wedge \alpha \simeq \gamma \wedge \beta $}.
\end{enumerate}
\end{lemma}
\begin{proof}{\ }
\begin{enumerate}
\item 
Assume \mbox{$ \alpha \vdash $}.
By {\bf EWeakening},  we have 
\mbox{$ \alpha , \neg {\bf 0} \vdash $}, i.e., \mbox{$ \alpha \rightarrow {\bf 0} $}.
But, by {\bf 0Axiom}, \mbox{$ {\bf 0} \vdash $} and, by {\bf EWeakening}, 
\mbox{$ {\bf 0} , \neg \alpha \vdash $} and \mbox{$ {\bf 0} \rightarrow \alpha $}.
We have shown the \emph{only if} direction.

Now, if \mbox{$ \alpha \rightarrow {\bf 0} $} we have, 
\mbox{$ \alpha , \neg {\bf 0} \vdash $}.
But, by {\bf NAxiom} and {\bf EWeakening} we have 
\mbox{$ \alpha , {\bf 0} \vdash $}. 
By {\bf Cut} we have \mbox{$ \alpha \vdash $}.
 \item 
\[ \begin{array}{lc}
{\bf Assumption} & \alpha , \neg \beta \vdash \\
\cline{2-2}
{\bf D} & \neg \neg \alpha , \neg \beta \vdash \\
\cline{2-2}
{\bf Exchange} & \neg \beta , \neg \neg \alpha \vdash 
\end{array} \]

\item
\[ \begin{array}{lclc}
& & {\bf NAxiom} & \beta \wedge \gamma , \neg ( \beta \wedge \gamma ) \vdash \\
\cline{4-4}
{\bf Assumption} & \alpha , \neg \beta \vdash
& {\bf LL-\wedge} & \beta , \gamma , 
\neg ( \beta \wedge \gamma ) \vdash \\
\cline{2-2} \cline{4-4}
{\bf Exchange} & \neg \beta , \alpha \vdash 
& {\bf Circ} & \neg ( \beta \wedge \gamma ) , \gamma , \beta \vdash \\
\cline{2-2} \cline{4-4}
{\bf EWeakening} & \neg ( \beta \wedge \gamma ) , \gamma , \neg \beta , \alpha \vdash
& {\bf EWeakening} &  \neg ( \beta \wedge \gamma ) , \gamma , \beta , \alpha  \vdash \\
\cline{2-4}
& {\bf Cut} &  \neg ( \beta \wedge \gamma ) , \gamma , \alpha \vdash \\
\cline{3-3}
& {\bf Circ} & \alpha , \gamma , \neg ( \beta \wedge \gamma ) \vdash \\
\cline{3-3}
& {\bf LL-\wedge} & \alpha \wedge \gamma , 
\neg ( \beta \wedge \alpha ) \vdash
\end{array} \]

\item
Assuming \mbox{$ \alpha \simeq \beta $}, the following derivation shows
that \mbox{$ \gamma \wedge \beta \rightarrow \gamma \wedge \alpha $}.
The converse implication is proved similarly.
\[ \begin{array}{cc}
{\bf NAxiom} 
& \gamma \wedge \alpha , \neg ( \gamma \wedge \alpha ) \vdash \\
\cline{2-2}
{\bf LL-\wedge} 
& \gamma , \alpha , \neg ( \gamma \wedge \alpha ) \vdash \\
\cline{2-2}
Lemma~\ref{the:equivalence} & 
\gamma , \beta , \neg ( \gamma \wedge \alpha ) \vdash \\
\cline{2-2}
{\bf LL-\wedge} & \gamma \wedge \beta , \neg ( \gamma \wedge \alpha ) \vdash
\end{array} \]
\end{enumerate}
\end{proof}

\begin{lemma} \label{the:dot-right}
For any \mbox{$ \alpha , \beta \in {\cal L} $}, 
\mbox{$ \alpha \wedge \beta \rightarrow \beta $}.
\end{lemma}
\begin{proof}
\[ \begin{array}{lr}
{\bf NAxiom} & \beta , \neg \beta \vdash \\
\cline{2-2}
{\bf EWeakening} & \alpha , \beta , \neg \beta \vdash \\
\cline{2-2}
{\bf LL-\wedge} & \alpha \wedge \beta , \neg \beta \vdash
\end{array} \]
\end{proof}

\begin{lemma} \label{the:3bc}
For any \mbox{$ \alpha , \beta , \gamma \in {\cal L} $} such that 
\mbox{$ \alpha \rightarrow \beta $}, 
\begin{enumerate}
\item \label{absorption}
\mbox{$ \alpha \wedge \beta \simeq \alpha \simeq \beta \wedge \alpha $}.
\item \label{g-absorption}
\mbox{$ \gamma \wedge ( \alpha \wedge \beta ) \simeq \gamma \wedge \alpha 
\simeq \gamma \wedge ( \beta \wedge \alpha ) $}.
\item \label{rule-assoc-b}
\mbox{$ ( \alpha \wedge \beta ) \wedge \gamma \simeq $}
\mbox{$ \alpha \wedge ( \beta \wedge \gamma ) $}.
\item \label{rule-assoc-c}
\mbox{$ ( \gamma \wedge \beta ) \wedge \alpha \simeq \gamma \wedge ( \beta \wedge \alpha ) $}.
\end{enumerate}
\end{lemma}
\begin{proof}
\begin{enumerate}
\item 
By Lemma~\ref{the:dot-right}, 
\mbox{$ \beta \wedge \alpha \rightarrow \alpha $}.
By Lemma~\ref{the:congruence}, item~\ref{cdot-left}, the assumption implies that
\mbox{$ \alpha \wedge \alpha \rightarrow \beta \wedge \alpha $}.
The reader will easily show that \mbox{$ \alpha \rightarrow \alpha \wedge \alpha $} 
and conclude that \mbox{$ \alpha \rightarrow \beta \wedge \alpha $}.
We have shown that \mbox{$ \beta \wedge \alpha \simeq \alpha $}.
The following derivation shows that 
\mbox{$ \alpha \wedge \beta \rightarrow \alpha $}.
\[ \begin{array}{rlrl}
{\bf Assumption} & \alpha , \neg \beta \vdash
& {\bf NAxiom} & \alpha , \neg \alpha \vdash \\
\cline{2-4}
& {\bf MWeakening} & \alpha , \beta , \neg \alpha \vdash \\
\cline{3-3}
& {\bf LL-\wedge} & \alpha \wedge \beta , \neg \alpha \vdash
\end{array} \]
Now, let us show that \mbox{$ \alpha \rightarrow \alpha \wedge \beta $}.
\[ \begin{array}{rlrl}
{\bf Assumption} & \alpha , \neg \beta \vdash 
& {\bf NAxiom} & \alpha , \neg \alpha \vdash \\
\cline{2-4}
& {\bf MWeakening} & \alpha , \beta , \neg \alpha \vdash \\
\cline{3-3}
& {\bf LL-\wedge} & \alpha \wedge \beta , \neg \alpha \vdash \\
\cline{3-3}
& {\bf Lemma~\ref{the:congruence}, item~\ref{negation}} 
& \neg \neg \alpha , \neg ( \alpha \wedge \beta ) \vdash \\
\cline{3-3}
& {\bf D} & \alpha , \neg ( \alpha \wedge \beta ) \vdash
\end{array} \]
\item 
By item~\ref{absorption} just above and Lemma~\ref{the:congruence}, 
item~\ref{cdot-right}.
\item 
Consider the following derivation:
\[ \begin{array}{rlrlll}
{\bf Assumption} & ( \alpha \wedge \beta ) \wedge \gamma , \Delta \vdash \\
\cline{2-2}
{\bf LL-\wedge} & \alpha , \beta , \gamma , \Delta \vdash 
& {\bf Assumption} & \alpha , \neg \beta \vdash \\
\cline{2-4}
{\bf Cut} & \alpha , \gamma , \Delta \vdash 
& & {\bf Assumption} & \alpha , \neg \beta \vdash \\
\cline{2-5}
& {\bf ML-\wedge} & \alpha , \beta \wedge \gamma , \Delta \vdash \\
\cline{3-3}
& {\bf LL-\wedge} & \alpha \wedge ( \beta \wedge \gamma ) , \Delta \vdash
\end{array} \]
By Lemma~\ref{the:equivalence}, we have shown that, under our assumption,
\mbox{$ (\alpha \wedge \beta ) \wedge \gamma \rightarrow $}
\mbox{$ \alpha \wedge ( \beta \wedge \gamma ) $}.

Consider, now
\[ \begin{array}{rlrll}
{\bf Assumption} & \alpha \wedge ( \beta \wedge \gamma ) , \Delta \vdash 
& {\bf Assumption} & \alpha , \neg \beta \vdash \\
\cline{2-2} \cline{4-4}
{\bf LL-\wedge} & \alpha , \beta \wedge \gamma , \Delta \vdash 
& {\bf EWeakening} & \alpha , \neg \beta , \Delta \vdash \\
\cline{2-4}
{\bf ML-\wedge} & \alpha , \gamma , \Delta \vdash 
& & {\bf Assumption} & \alpha , \neg \beta \vdash \\
\cline{2-5}
& {\bf MWeakening} & \alpha , \beta , \gamma , \Delta \vdash \\
\cline{3-3}
& {\bf LL-\wedge} & ( \alpha \wedge \beta ) \wedge \gamma , \Delta \vdash
\end{array} \]
We have shown \mbox{$ \alpha \wedge ( \beta \wedge \gamma ) \rightarrow $}
\mbox{$ (\alpha \wedge \beta ) \wedge \gamma $}.

\item 
Consider
\[ \begin{array}{rlrl}
{\bf NAxiom} & \gamma \wedge \alpha , \neg ( \gamma \wedge \alpha ) \vdash \\
\cline{2-2}
{\bf Exchange} & \neg ( \gamma \wedge \alpha ) , \gamma \wedge \alpha \vdash 
& {\bf Assumption} & \alpha , \neg \beta \vdash \\
\cline{2-2} \cline{4-4}
{\bf LR-\wedge} & \neg ( \gamma \wedge \alpha ) , \alpha , \gamma \vdash
& {\bf EWeakening} & \neg ( \gamma \wedge \alpha ) , \alpha , \neg \beta \vdash \\
\cline{2-4} 
& {\bf M-Weakening} & \neg ( \gamma \wedge \alpha ) , \alpha , \beta , \gamma \vdash  \\
\cline{3-3} 
& {\bf LR-\wedge} & \neg ( \gamma \wedge \alpha ) , \alpha , \gamma \wedge \beta \vdash \\
\cline{3-3}
& {\bf Circ} & \gamma \wedge \beta , \alpha , \neg ( \gamma \wedge \alpha ) 
\vdash \\
\cline{3-3}
& {\bf LL-\wedge} & ( \gamma \wedge \beta ) \wedge \alpha , 
\neg ( \gamma \wedge \alpha ) \vdash \\
\end{array} \]
We have shown that 
\mbox{$ ( \gamma \wedge \beta ) \wedge \alpha \rightarrow \gamma \wedge \alpha $}.
Let us show the inverse implication.
\end{enumerate}
\[\footnotesize \begin{array}{rlrl}
{\bf NAxiom} & ( \gamma \wedge \alpha ) \wedge \beta , 
\neg ( ( \gamma \wedge \beta ) \wedge \alpha ) \vdash 
& {\bf Assumption} & \alpha , \neg \beta \vdash \\
\cline{2-2} \cline{4-4}
{\bf LL-\wedge} & \gamma , \alpha , \beta , \neg ( ( \gamma \wedge \beta ) \wedge \alpha ) \vdash 
& {\bf EWeakening} & \gamma , \alpha , \neg \beta , \neg ( ( \gamma \wedge \beta ) \wedge \alpha ) \vdash \\
\cline{2-4}
& {\bf Cut} & \gamma , \alpha , \neg ( ( \gamma \wedge \beta ) \wedge \alpha ) \vdash \\
\cline{3-3}
& {\bf LL-\wedge} & \gamma \wedge \alpha , \neg ( ( \gamma \wedge \beta ) \wedge \alpha ) \vdash 
\end{array} \]

We have shown that 
\mbox{$ ( \gamma \wedge \beta ) \wedge \alpha \simeq \gamma \wedge \alpha $}.

Now, by item~\ref{g-absorption} above, we have 
\mbox{$ \gamma \wedge ( \alpha \wedge \beta ) \simeq $}
\mbox{$ \gamma \wedge \alpha $}.
\end{proof}

\subsection{Completeness} \label{sec:completeness}
We can move now to the completeness result.
Let us denote by $X$ the set all equivalence classes in $\cal L$ over logical equivalence: 
\mbox{$ X = {\cal L} / \! \! \simeq $}.
For any \mbox{$ \alpha \in {\cal L} $}, 
\mbox{$ \overline{\alpha} \in X $} is the equivalence class of $\alpha$ under $\simeq$. 
Our goal is  now to show that $X$ can be equipped with a P-algebra structure.
We need to define a $0$ element, a complement operation and a ``$\cdot$'' operation
on $X$ that satisfy the conditions of Definition~\ref{def:P-algebra}.

Lemma~\ref{the:congruence} has just shown that the logical operations 
can be considered to operate on the equivalence classes under $\simeq$, i.e. on $X$ 
and this enables us to define the structure that we want to consider in the 
completeness proof.

\begin{definition} \label{def:logicalP}
Let \mbox{$ X = {\cal L} / \! \! \simeq $}.
The element $0$ of $X$ is defined by 
\mbox{$ 0 = \overline{\bf 0}$}.
Note that, by Lemma~\ref{the:congruence}, item~\ref{contradiction}, 
\mbox{$ 0 = \{ \alpha \in {\cal L} \mid \alpha \vdash \} $}.
For any \mbox{$ x \in X $}, \mbox{$ x' \eqdef \overline{\neg \alpha} $}
for any \mbox{$ \alpha \in x $}.
For any \mbox{$ x , y \in X $}, 
\mbox{$ x \cdot y \eqdef \overline{\alpha \wedge \beta} $}
for any \mbox{$ \alpha \in x $}, \mbox{$ \beta \in y $}.
\end{definition}
We want to show that the structure \mbox{$ P = \langle X , 0 , ' , \cdot \rangle $} is
a P-algebra. 

\begin{lemma} \label{the:leq1}
For any \mbox{$ x , y \in X $}, \mbox{$ x \leq y $} iff 
\mbox{$ \alpha \rightarrow \beta $} for any 
\mbox{$ \alpha \in x $} and \mbox{$ \beta \in y $},  
equivalently, iff \mbox{$ \alpha \rightarrow \beta $} for some 
\mbox{$ \alpha \in x $} and some \mbox{$ \beta \in y $}.
Therefore the algebra $P$ satisfies {\bf Partial Order} in Definition~\ref{def:basic}.
\end{lemma}
\begin{proof}
Let \mbox{$ \alpha \in x $} and \mbox{$ \beta \in y $}.
We have \mbox{$ \alpha \wedge \beta \in x \cdot y $}.
If \mbox{$x \leq y$}, \mbox{$ x \cdot y = x $}, 
\mbox{$ \alpha \wedge \beta \in x $}, 
\mbox{$ \alpha \wedge \beta \simeq \alpha $}.
Since \mbox{$ \alpha \wedge \beta \rightarrow \beta $} 
by Lemma~\ref{the:dot-right}, we have \mbox{$ \alpha \rightarrow \beta $}.
If, now \mbox{$ \alpha \rightarrow \beta $}, 
\mbox{$ \alpha \wedge \beta \simeq \alpha $} by Lemma~\ref{the:3bc}, 
item~\ref{absorption} and \mbox{$ x \cdot y = x $}.
Lemma~\ref{the:implication} shows that $\leq$ is a partial order.
\end{proof}

\begin{lemma} \label{the:smile1}
\mbox{$ x \smile y $} iff \mbox{$ \alpha \wedge \beta \rightarrow \alpha $}
for any \mbox{$ \alpha \in x $} and \mbox{$ \beta \in y $}.
The relation $\smile$ is symmetric.
\end{lemma}
\begin{proof}
The first claim follows from Lemma~\ref{the:leq1}.
For the second claim, assume that \mbox{$ \alpha \wedge \beta \rightarrow \alpha $}.
Let us show that \mbox{$ \beta \wedge \alpha \rightarrow \beta $}.

\[\footnotesize \begin{array} {rlrlrl}
{\bf Assumption} & \alpha \wedge \beta , \neg \alpha \vdash \\
\cline{2-2} 
{\bf LL-\wedge} & \alpha , \beta , \neg \alpha \vdash 
& {\bf Assumption} & \alpha \wedge \beta , \neg \alpha \vdash 
& {\bf Lemma~\ref{the:dot-right}} & \alpha \wedge \beta , \neg \beta \vdash \\
\cline{2-2} \cline{4-6}
{\bf Circ} & \neg \alpha , \beta , \alpha \vdash 
& & {\bf MWeakening} & \alpha \wedge \beta , \alpha , \neg \beta \vdash \\
\cline{2-2} \cline{5-5}
{\bf EWeakening} & \neg \alpha , \beta , \alpha , \neg \beta \vdash 
& & {\bf LL-\wedge} & \alpha , \beta , \alpha , \neg \beta \vdash \\
\cline{2-4}
& {\bf Cut} & \beta , \alpha , \neg \beta \vdash \\
\cline{3-3}
& {\bf LL-\wedge} & \beta \wedge \alpha , \neg \beta \vdash
\end{array} \]
\end{proof}

\begin{lemma} \label{the:associativity}
The algebra $P$ satisfies the three {\bf P-associativity} properties.
\end{lemma}
\begin{proof}{\ }
\begin{enumerate}
\item
\mbox{$ ( \alpha \wedge \beta ) \wedge \gamma \simeq {\bf 0} $}
iff \mbox{$ ( \alpha \wedge \beta ) \wedge \gamma \vdash $} by 
Lemma~\ref{the:congruence}, item~\ref{contradiction}, iff, by {\bf Circ}, 
{\bf $\wedge$Left-Elim} and {\bf LL-$\wedge$}, 
\mbox{$ ( \gamma \wedge \beta ) \wedge \alpha \vdash $}.
\item
By Lemma~\ref{the:3bc}, item~\ref{rule-assoc-b}.
\item
By Lemma~\ref{the:3bc}, item~\ref{rule-assoc-c}.
\end{enumerate}
\end{proof}

\begin{lemma} \label{the:dot-monotonicity}
The algebra $P$ satisfies the two {\bf Dot-monotonicity} properties.
\end{lemma}
\begin{proof}
Property~\ref{left-monotonicity} follows from Lemma~\ref{the:congruence}, 
item~\ref{cdot-left}.
Property~\ref{right-monotonicity} follows from\break Lemma~\ref{the:dot-right}.
\end{proof}

\begin{lemma} \label{the:Z}
The algebra $P$ satisfies property {\bf Z}.
\end{lemma}
\begin{proof}
By {\bf 0Axiom}, \mbox{$ {\bf 0} \vdash $}.
By {\bf EWeakening}, \mbox{$ {\bf 0} , \neg \alpha \vdash $}, i.e., 
\mbox{$ {\bf 0} \rightarrow \alpha $} and \mbox{$ {\bf 0} \leq \alpha $}.
\end{proof}

\begin{lemma} \label{the:Comp}
The algebra $P$ satisfies property {\bf Comp}.
\end{lemma}
\begin{proof}
By {\bf NAxiom}, \mbox{$ \alpha , \neg \alpha \vdash \simeq {\bf 0} $}.
\end{proof}

\begin{lemma} \label{the:O}
The algebra $P$ satisfies property {\bf O}.
\end{lemma}
\begin{proof}
Assume \mbox{$ \alpha \wedge \beta \rightarrow \gamma $} and 
\mbox{$ \alpha \wedge \neg \beta \rightarrow \gamma $}.
Consider the derivation:
\[\begin{array}{rlrl}
{\bf Assumption} & \alpha \wedge \beta , \neg \gamma \vdash 
& {\bf Assumption} & \alpha \wedge \neg \beta , \neg \gamma \vdash \\
\cline{2-2} \cline{4-4}
{\bf LL-\wedge} & \alpha , \beta , \neg \gamma \vdash
& {\bf LL-\wedge} & \alpha , \neg \beta , \neg \gamma \vdash \\
\cline{2-4}
& {\bf Cut} & \alpha , \neg \gamma \vdash 
\end{array}\]
We have shown that if 
\mbox{$ \overline{\alpha} \cdot \overline{\beta} \leq \overline{\gamma} $}
and \mbox{$ \overline{\alpha} \cdot \overline{\beta}' \leq \overline{\gamma} $}, 
one has \mbox{$ \overline{\alpha} \leq \overline{\gamma} $}.
\end{proof}

We may now summarize.
\begin{theorem} \label{the:rules-P-alebra}
The logical algebra $P$ defined in Definition~\ref{def:logicalP} is a P-algebra.
\end{theorem}
\begin{proof}
By Lemmas~\ref{the:leq1} to~\ref{the:O}.
\end{proof}

Note that the sequent \mbox{$ \vdash \alpha \vee \neg \alpha $} 
is valid, i.e., quantum logic satisfies the Law of Excluded Middle.
Nevertheless quantum logic is not bivalent: even an \emph{atomic feature} 
(see Section~\ref{sec:atomic}) $a$ of a P-algebra may be such that 
\mbox{$ a \not \leq x $} and \mbox{$ a \not \leq x' $}.

We may now prove a completeness result.
\begin{theorem} \label{the:completeness}
Let us define an assignment of features of $P$ to any atomic proposition by:
\mbox{$ v( \sigma ) \eqdef \overline{\sigma} $} for any \mbox{$ \sigma \in AT $}.
For any sequence $\Gamma$ of propositions of $ \cal L$:
\mbox{$ \Gamma \vdash $} iff \mbox{$ \Gamma  \models_{v} 0 $}.
\end{theorem}
\begin{proof}
\mbox{$ \Gamma \vdash $} iff \mbox{$ \phi( \Gamma ) \vdash $}
by {\bf LL-$\wedge$} and {\bf $\wedge$Left-Elim}, iff
\mbox{$ \phi( \Gamma ) \simeq {\bf 0} $} by Lemma~\ref{the:congruence}, 
item~\ref{contradiction}, iff \mbox{$ v( \phi( \Gamma ) ) = {\bf 0} $}, iff
\mbox{$ v( \phi( \Gamma ) ) \leq  {\bf 0} $}, iff 
\mbox{$ \Gamma \models_{v} $} by Equation~(\ref{eq:sequent}).
\end{proof}

\section{Conclusion and further research} \label{sec:further}
This paper puts forward a family of algebraic structures, 
some sort of non-commut- ative Boolean algebras called P-algebras, in which the 
conjunction has the properties of the projection operation between subspaces of
an inner-product vector space.
They are orthomodular ordered structures but not lattices.
It claims that such structures ought to be for Quantum Logic what Boolean algebras are 
for Classical Logic.
It supports this claim by a complete characterization of the logic offered by P-algebras.

A large number of different avenues for further research open naturally and should
attract different communities.

Proof-theorists could be interested in studying the properties of system $\cal R$ or
equivalent systems and consider Cut-elimination and normalization. 
The question of the number of non-equivalent propositions on a finite number of
atomic propositions also seems interesting.

Algebraists could be interested in a representation result for P-algebras, parallel to
Stone's representation theorem.
The study of atomic P-algebras, initiated in Section~\ref{sec:atomic}, should be 
deepened and, for example, the question of a natural topology on atomic features 
should be considered.
Boolean algebras are deeply related to Boolean rings, what is the connection between
P-algebras and non-commutative Boolean rings ( they probably should not be 
associative either).

What are the types of morphisms one should consider between P-algebras?
Which ones could be considered a reasonable generalization of linear maps?
Is there a tensor product in the category of P-algebras?
Can it throw some light on quantic entanglement?

The most urgent research direction is certainly to consider richer algebraic structures
in which one could model quantum probabilities and differential equations such as 
Schr\"{o}dinger's.
The mathematical structures used by quantum physicists are countably based atomic 
P-algebras (see Appendices B and C) equipped by a function of type 
\mbox{$ A(X) \times A(X) \rightarrow [0 , 1] $} that represents 
the \emph{transition probability} between pure states.
This transition probability is symmetric and satisfies a fundamental property noticed 
in Theorem 1 of~\cite{Qsuperp:IJTP}.
At this point it seems that not all P-algebras can be equipped by such a function and
that, for those that are equipped by such a function, a feature can be identified with 
the set of its atomic features.
The relation of such functions to probabilities seems to be an intriguing topic for further 
research.

\section*{Acknowledgements}
My deepest thanks to Kurt Engesser who directed me to John von Neumann's 
letter quoted in Section~\ref{sec:doubts} that convinced me that my intuitions
were worth pursuing.
\bibliographystyle{plain}

\begin{thebibliography}{10}

\bibitem{BirkvonNeu:36}
Garrett Birkhoff and John von Neumann.
\newblock The logic of quantum mechanics.
\newblock {\em Annals of Mathematics}, 37:823--843, 1936.

\bibitem{Buszkowski:2017}
W.~Buszkowski.
\newblock Involutive nonassociative lambek calculus: Sequent systems and
  complexity.
\newblock {\em Bulletin of the Section of Logic}, 46:75--91, 2017.

\bibitem{Chajda+2:2018}
Ivan Chajda, Helmut L\"{a}nger, and Jan Paseka.
\newblock Uniquely complemented posets.
\newblock {\em Order}, 35:421--431, 2018.

\bibitem{Cintula+2:2013}
P.~Cintula, R.~Hor\^{c}\'{i}k, and C.~Noguera.
\newblock Nonassociative substructural logics and their semilinear extensions:
  Axiomatization and completeness properties.
\newblock {\em The Review of Symbolic Logic}, 6:394--423, 2013.

\bibitem{Fazio+3:2021}
Davide Fazio, Antonio Leda, Francesco Paoli, and Gavin~St. John.
\newblock A substructural gentzen calculus for orthomodular quantum logic.
\newblock {\em The Review of Symbolic Logic}, January 2022.
\newblock DOI: https://doi.org/10.1017/S1755020322000016.

\bibitem{Finch_lattice:69}
P.~D. Finch.
\newblock On the lattice structure of quantum logic.
\newblock {\em Bulletin of the Australian Mathematical Society}, 1:333--340,
  1969.

\bibitem{Galatos+1:2009}
N.~Galatos and H.~Ono.
\newblock Cut elimination and strong separation for substructural logics: An
  algebraic approach.
\newblock {\em Annals of Pure and Applied Logic}, 161:1097--1133, 2009.

\bibitem{Gent:32}
Gerhard Gentzen.
\newblock \"{U}ber die existenz unabh\"{a}ngiger axiomensysteme zu unendlichen
  satzsystemen.
\newblock {\em Mathematische Annalen}, 107:329--350, 1932.

\bibitem{Gent:69}
Gerhard Gentzen.
\newblock {\em The Collected Papers of Gerhard Gentzen, edited by M. E. Szabo}.
\newblock North Holland, Amsterdam, 1969.

\bibitem{MacLane:1938}
Saunders~Mac Lane.
\newblock A lattice formulation for transcendence degree and p-bases.
\newblock {\em Duke Mathematical Journal}, 4(3,38-00438-7):455--468, September
  1938.

\bibitem{Connectives:Fest}
Daniel Lehmann.
\newblock Connectives in cumulative logics.
\newblock In {\em Pillars of Computer Science, Essays dedicated to Boris (Boaz)
  Trakhtenbrot on the occasion of his 85th birthday}, number 4800 in Lecture
  Notes in Computer Science, pages 424--440. Springer Verlag, 2008.

\bibitem{Lehmann_andthen:JLC}
Daniel Lehmann.
\newblock A presentation of quantum logic based on an "and then" connective.
\newblock {\em Journal of Logic and Computation}, 18(1):59--76, February 2008.
\newblock doi: 10.1093/logcom/exm054.

\bibitem{Qsuperp:IJTP}
Daniel Lehmann.
\newblock Quantic superpositions and the geometry of complex {H}ilbert spaces.
\newblock {\em International Journal of Theoretical Physics}, 47(5):1333--1353,
  May 2008.
\newblock DOI:10.1007/s10773-007-9576-y.

\bibitem{SP:IJTP}
Daniel Lehmann.
\newblock Similarity-projection structures: the logical geometry of quantum
  physics.
\newblock {\em International Journal of Theoretical Physics}, 48(1):261--281,
  2009.
\newblock DOI: 10.1007/s10773-008-9801-3.

\bibitem{Lehmann-metalinear1:2022}
Daniel Lehmann.
\newblock Metalinear structures and the substructural logic of quantum
  measurements.
\newblock http://arxiv.org/abs/2201.02043v1, December 2022.

\bibitem{Lehmann-P-algebras-v1:2024}
Daniel Lehmann.
\newblock Projection-algebras and quantum logic.
\newblock https://doi.org/10.48550/arXiv.2402.07042, submitted to a journal,
  February 2024.

\bibitem{LEG:Malg}
Daniel Lehmann, Kurt Engesser, and Dov~M. Gabbay.
\newblock Algebras of measurements: the logical structure of quantum mechanics.
\newblock {\em International Journal of Theoretical Physics}, 45(4):698--723,
  April 2006.
\newblock DOI 10.1007/s10773-006-9062-y.

\bibitem{Moot+1:2012}
R.~Moot and C.~Retor\'{e}.
\newblock The non-associative lambek calculus.
\newblock In {\em The Logic of Categorial Grammars}, pages 101--147. Springer,
  Berlin, 2012.

\bibitem{PutnamHow:1974}
Hilary Putnam.
\newblock How to think quantum logically.
\newblock {\em Synthese}, 29:55--61, 1974.

\bibitem{vNeumann_letters}
Mikl\'{o}s R\'{e}dei, editor.
\newblock {\em John von Neumann: Selected Letters}, volume~27 of {\em History
  of Mathematics}.
\newblock London Mathematical Society - American Mathematical Society, 2005.

\bibitem{Steinitz:1930}
Ernst Steinitz.
\newblock {\em Algebraische Theorie der K\"{o}rper}.
\newblock 1930.

\bibitem{vonNeumann:Quanten}
John von Neumann.
\newblock {\em Mathematische Grundlagen der Quanten-mechanik}.
\newblock Springer Verlag, Heidelberg, 1932.
\newblock American edition: Dover Publications, New York, 1943.

\bibitem{Waphare+1:2005}
B.~N. Waphare and V.~V. Joshi.
\newblock On uniquely complemented posets.
\newblock {\em Order}, 22:11--20, 2005.

\end{thebibliography}

\appendix
\section{Commuting features} \label{sec:commuting}
Our next results concern commuting features.
\begin{definition} \label{def:c-set}
A set of features \mbox{$ Y \subseteq X $} is a \emph{c-set} iff any pair of features
of $Y$ commutes.
\end{definition}

\begin{lemma} \label{the:commuting-features}
For any \mbox{$x , y \in X $} such that \mbox{$ x \smile y $}, equivalently 
\mbox{$ x \cdot y = y \cdot x $} 
\begin{enumerate}
\item \label{glb-lub}
$x \cdot y$ is the g.l.b. of $x$ and $y$, and $x + y$ is their l.u.b.
Therefore the operations $\cdot$ and $+$ between orthogonal features are associative 
and commutative.
\item \label{boolean-algebra}
any c-set that includes $0$ and is closed under $'$ is a boolean algebra.
\end{enumerate}
\end{lemma}
\begin{proof}
\begin{enumerate}
\item 
Since $x$ and $y$ commute, by Theorem~\ref{the:basic1}, item~\ref{smile-com}, 
\mbox{$ x \cdot y \leq x $} and \mbox{$ x \cdot y \leq y $}. 
One concludes by {\bf Dot-monotonicity}, item~\ref{left-monotonicity}.
The last claim follows by duality.
\item 
By item~\ref{glb-lub} just above, the structure is a lattice and one easily verifies all the
necessary properties.
\end{enumerate}
\end{proof}

\section{Countably based P-algebras} \label{sec:CBP}
A central property of P-algebras is that the partial order relation $\leq$ does not, 
in general, equip the carrier $X$ with a lattice structure, since \mbox{$ x \cdot y $} is
not, in general, a lower bound for $x$ and \mbox{$ x + y $} is not an upper bound for
$y$.
Nevertheless, as shown in Theorem~\ref{the:basic1}, item~\ref{smile-glb} and 
Theorem~\ref{the:co+}, item~\ref{+lub}, if \mbox{$ x \smile y $}, then
\mbox{$ x \cdot y = $} \mbox{$ g.l.b. ( x , y ) $} and 
\mbox{$ x + y = $} \mbox{$ l.u.b. ( x , y ) $}.
In particular, if \mbox{$ x \, \bot \, y $}, then \mbox{$ x + y = $} \mbox{$ y + x = $} 
\mbox{$ l.u.b. ( x , y ) $}.

To deal with Quantum Physics one needs to be able to consider infinite dimensional 
spaces (or features). 
We shall prove the equivalence of two properties and define countably based 
P-algebras.
\begin{definition} \label{def:ortho-set}
Given a P-algebra \mbox{$ \langle X , 0 , ' , \cdot \rangle $}, 
a set \mbox{$ Y \subseteq X $} of features is said to be an \emph{ortho-set} 
iff any two different elements of $Y$ are orthogonal: for any \mbox{$ x , y \in Y $} 
such that \mbox{$ x \neq y $}, one has \mbox{$ x \, \bot \, y $}.
\end{definition}

\begin{theorem} \label{the:CB}
The following properties are equivalent:
\begin{enumerate}
\item \label{ortho-CB}
any countable ortho-set has a l.u.b.
\item \label{sequence-CB}
any ascending sequence of features \mbox{$ x_{0} \leq x_{1} \leq \ldots $} 
has a l.u.b.
\end{enumerate}
\end{theorem}
\begin{proof}
Assume, first, that any countable ortho-set has a l.u.b. and let 
\mbox{$ \{ x_{i} \}_{i \in {\cal N}} $} be any ascending sequence of features.
For any \mbox{$ i \in {\cal N} $}, \mbox{$ x_{i} \leq x_{i+1} $} and, therefore, 
\mbox{$ x_{i} \smile x_{i+1} $}, \mbox{$ x_{i+1} \smile x_{i} $} and
\mbox{$ x_{i+1} \smile x_{i}' $}, i.e., \mbox{$ x_{i+1} \cdot x_{i}' \leq x_{i+1} $}.
Let \mbox{$ y_{0} = x_{0} $} and \mbox{$ y_{i+1} = $}
\mbox{$ x_{i+1} \cdot x_{i}' $}.
For any \mbox{$ i \in {\cal N} $}, \mbox{$ y_{i} \leq x_{i} $}.
By {\bf Dot-monotonicity}, item~\ref{right-monotonicity}, 
\mbox{$ y_{i+1} \, \bot \, x_{i} $} and, since \mbox{$ y_{k} \leq x_{i} $} for any 
\mbox{$ k \leq i $}, \mbox{$ y_{i+1} \, \bot \, y_{k} $} for any such $k$.
Therefore \mbox{$ Y = \{ y_{i} \}_{i \in {\cal N}} $} is an ortho-set.
Let \mbox{$ y = l.u.b. ( Y ) $}.
By Theorem~\ref{the:orthomodularity}
\begin{equation} \label{eq:xi}
x_{i+1} = x_{i} + x_{i+1} \cdot x_{i}' = x_{i} + y_{i+1}.
\end{equation}
One easily shows that, for any \mbox{$ i \in {\cal N} $}, 
\mbox{$ x_{i} = \sum_{0 \leq k \leq i} y_{i} $}.
We have shown that, for any \mbox{$ i \in {\cal N} $}, 
\mbox{$ y_{i} \leq x_{i} \leq \sum_{0 \leq k \leq i} y_{i} $}.
We see that \mbox{$ y = l.u.b. ( \{ x_{i} \}_{i \in {\cal N}} ) $}.

Assume, now that any ascending sequence of features has a l.u.b. and let 
\mbox{$ Y = \{ y_{i} \}_{i \in {\cal N}} \subseteq X $} be a countable ortho-set.
Let \mbox{$ x_{i} = \sum_{0 \leq k \leq i} y_{k} $}.
The sequence \mbox{$ X = \{ x_{i} \}_{i \in {\cal N}} $} is ascending. 
Let \mbox{$ x = l.ub. ( X ) $}.
Since, for any $i$, \mbox{$ y_{i} \leq x_{i} $}, $x$ is an upper bound for $Y$.
But, by Theorem~\ref{the:co+}, item~\ref{+ub}, 
if $y$ is an upper bound for $Y$, then \mbox{$ x_{i} \leq y $} for any $i$ and
\mbox{$ x \leq y $}.
\end{proof}

\begin{definition} \label{def:CBPA}
A P-algebra that satisfies the properties of Theorem~\ref{the:CB} is said to be 
\emph{countably based}.
\end{definition}

\section{Atomic P-algebras} \label{sec:atomic}
Atomic Boolean algebras present a particularly interesting family of Boolean algebras.
We shall now study atomic P-algebras.
This is particularly important because the Hilbert spaces that form the setting of QM, 
when considered as P-algebras as we did, are atomic P-algebras.
The atoms are the one-dimensional subspaces and they model the \emph{pure states} 
of the quantic system, a central concept in QM.
The atoms of a Boolean algebra may be defined in many different but equivalent ways, 
but properties equivalent in Boolean algebras are not, in general, equivalent 
in P-algebras.

\begin{definition} \label{def:atom}
Given a P-algebra \mbox{$ P = \langle X , 0 , ' , \cdot \rangle $}, 
a feature \mbox{$a \in X $} such that \mbox{$ a \neq 0 $} is an \emph{atomic} 
feature iff for any \mbox{$ x \in X $} either \mbox{$ x \, \bot \, a $} or 
\mbox{$ x \cdot a = a $}.
We shall denote the set of all atomic features by $\cal A$.
\end{definition}

The proof of the following lemma is left to the reader.
\begin{lemma} \label{the:atom}
A feature \mbox{$ a \in X $} is an atomic feature iff for any feature \mbox{$ x \in X $} 
such that \mbox{$ x \leq a $} either \mbox{$ x = 0 $} or \mbox{$ x = a $}.
\end{lemma}

Any feature \mbox{$ x \in X $} can be associated with the set of all atomic features
that imply it: 
\begin{equation} \label{eq:A()}
A ( x ) \eqdef \{ a \in {\cal A} \mid a \leq x \}.
\end{equation}

\begin{lemma} \label{the:at1}
In any P-algebra and for any \mbox{$ x , y \in X $} and any 
\mbox{$ a , b \in {\cal A} $}
\begin{enumerate}
\item \label{AxsubAy} 
If \mbox{$ x \leq y $}, then \mbox{$ A ( x ) \subseteq A ( y ) $}.
\item \label{A0}
\mbox{$ A( 0 ) = \emptyset $}.
\item \label{A'}
\mbox{$ A ( x' ) = \{ b \in {\cal A} \mid b \, \bot \, x \} $}.
\item \label{leq+-right}
\mbox{$ a = x \cdot a + x' \cdot a $}.
\item \label{acdotx}
If \mbox{$ a \not \! \! \bot \, x $}, then $a \cdot x$ is an atomic feature.
\item \label{Acdot}
For any \mbox{$ a \in A(x) $} such that \mbox{$ a \not \! \! \bot \, y $}, one has 
\mbox{$ a \cdot y \in A( x \cdot y ) $}.
\end{enumerate}
\end{lemma}
Note that the generalization of item~\ref{leq+-right} to arbitrary features:
\mbox{$ y = x \cdot y + x' \cdot y $} does not hold in P-algebras or in Hilbert spaces.
\begin{proof}
\begin{enumerate}
\item 
Assume \mbox{$ x \leq y $}.
Let \mbox{$ a \in A ( x ) $}, one has \mbox{$ a \leq x \leq y $}, 
\mbox{$ a \leq y $} and \mbox{$ a \in A ( y ) $}.
\item 
By Definition~\ref{def:atom}.
\item 
If \mbox{$ b \in A ( x' ) $}, \mbox{$ b \leq x' $} by Definition~\ref{def:atom}.
By Theorem~\ref{the:basic1}, item~\ref{bot'}, then, \mbox{$ b \, \bot \, x $}. 
If, now, \mbox{$ b \in {\cal A} $} and \mbox{$ b \, \bot \, x $} we have, 
by Theorem~\ref{the:basic1} item~\ref{bot'}, \mbox{$ b \leq x' $}, i.e., 
\mbox{$ b \in A( x' ) $}.
\item 
Since $a$ is atomic, both $x \cdot a$ and $x' \cdot a$ are either $0$ or $a$.
But, for any \mbox{$ y \in X $}, \mbox{$ y + y = y + 0 = 0 + y = y $}.
\item 
By our assumption \mbox{$ x \cdot a \neq 0 $} and by Definition~\ref{def:atom}, 
\mbox{$ x \cdot a = a $}.
For any \mbox{$ y \in X $} we have 
\mbox{$ y \cdot ( x \cdot a ) = $} \mbox{$ y \cdot a $}.
If \mbox{$ y \not \! \! \bot \, ( x \cdot a ) $}, then \mbox{$ y \, \bot \, a $} 
and \mbox{$ y \cdot a = a $} and therefore \mbox{$ y \cdot ( x \cdot a ) = x \cdot a $}.
\item 
By {\bf Dot-Monotonicity}, item~\ref{left-monotonicity}, \mbox{$ a \cdot y \leq x \cdot y $} 
and, by item~\ref{acdotx} just above \mbox{$ a \cdot y $} is an atomic feature.
\end{enumerate}
\end{proof}
Note that we do not claim that given an atomic feature \mbox{$ a \in A( x \cdot y ) $} there
is an atomic feature \mbox{$ b \in A(x) $} such that \mbox{$ a = b \cdot y $}.
At this point one may reasonably wonder whether the atomic features of a P-algebra 
satisfy the Mac Lane-Steinitz exchange property put in evidence by Ernst 
Steinitz~\cite{Steinitz:1930} in vector spaces and by Saunders 
MacLane~\cite{MacLane:1938} in matroids. 
The remark on the operation $+$ found after Definition~\ref{def:+} 
shows that it is not the case: one may find atomic features $a$, $b$ and $c$ such that 
\mbox{$ a \leq b + c $}, \mbox{$ a \not \leq b $} but \mbox{$ c \not \leq b + a $}.

We may now define an atomic P-algebra in a way that parallels the definition of an
atomic Boolean algebra and we shall show that, in an atomic P-algebra, a feature 
is characterized by the set of its atomic features.

\begin{definition} \label{def:at2}
A P-algebra is an \emph{atomic} P-algebra iff for any feature \mbox{$ x \in X $}, 
\mbox{$ x \neq 0 $}, \mbox{$ A( x ) \neq \emptyset $}.
\end{definition}

The following result expresses the central property of atomic P-algebras.
It is the basis of the Gram-Schmidt process.
\begin{lemma} \label{the:Gram-Schmidt}
Let \mbox{$ P = \langle X , 0 , ' , \cdot \rangle $} be an atomic P-algebra.
For any features \mbox{$ x , y \in X $} such that \mbox{$ x \leq y $} and
\mbox{$ x \neq y $}, there is some atomic feature \mbox{$ a \in {\cal A} $}
such that \mbox{$ a \leq y $} and \mbox{$ a \, \bot \, x $}.
\end{lemma}
\begin{proof}
By orthomodularity, Theorem~\ref{the:orthomodularity}, 
\mbox{$ y = $} \mbox{$ x + y \cdot x' = $} \mbox{$ y \cdot x' + x $} and the 
assumptions imply that \mbox{$ y \cdot x' \neq 0 $}.
Since $P$ is an atomic P-algebra, there is some atomic feature $a$ such that
\mbox{$ a \leq y \cdot x' $}.
By {\bf Dot-monotonicity}, item~\ref{right-monotonicity}, \mbox{$ a \leq x' $}
and, by Theorem~\ref{the:co+}, item~\ref{+left-monotonicity}, 
\mbox{$ a \leq y \cdot x' + x = y $}.
\end{proof}

Our last result concerns atomic P-algebras.
\begin{theorem} \label{the:atomic-Ax}
In an atomic P-algebra, for any \mbox{$ x \in X$}, any \mbox{$ a \in A(X) $}
\mbox{$ A( x ) = \emptyset $} iff \mbox{$ x = 0 $}.
\end{theorem}
\begin{proof}
The \emph{if} part is item~\ref{A0} in Lemma~\ref{the:at1}.
If \mbox{$ x > 0 $}, there exists an atomic feature \mbox{$ a \leq x $} 
by Definition~\ref{def:at2} and \mbox{$ a \in A(x) $}.
\end{proof}

\end{document}